\documentclass[11pt]{article}

\usepackage{times}
\usepackage{float}

\usepackage{amsfonts,amstext,amsmath,amssymb,amsthm}

\usepackage{chngpage}
\usepackage{rotating}
\usepackage{graphicx}
\usepackage{epstopdf} % to be loaded after the 'graphicx' package

\usepackage{caption}
\usepackage{subcaption}
\usepackage{tabularx}
\usepackage{booktabs}
\usepackage{multirow}

\newcolumntype{Y}{>{\centering\arraybackslash}X}

\newlength{\verticalcompensationlength}
\newcounter{verticalcompensationrows}

\usepackage{empheq}

\usepackage{paralist}

\long\def\symbolfootnote[#1]#2{\begingroup%
\def\thefootnote{\fnsymbol{footnote}}\footnote[#1]{#2}\endgroup}
\usepackage[%
    bookmarks=true,         % show bookmarks bar?
    unicode=false,          % non-Latin characters in Acrobat’s bookmarks
    pdftoolbar=true,        % show Acrobat’s toolbar?
    pdfmenubar=true,        % show Acrobat’s menu?
    hyperfootnotes=false,   % suppress hyperref's hyperfootnotes
    pdffitwindow=false,     % window fit to page when opened
    pdfstartview={FitH},    % fits the width of the page to the window
    pdftitle={An Accurate Method for Determining the Pre-Change Run-Length Distribution of the Generalized Shiryaev--Roberts Detection Procedure},    % title
    pdfauthor={Aleksey S. Polunchenko and Grigory Sokolov and Wenyu Du},     % author
    pdfsubject={An Accurate Method for Determining the Pre-Change Run-Length Distribution of the Generalized Shiryaev--Roberts Detection Procedure},   % subject of the document
    pdfcreator={Aleksey S. Polunchenko},   % creator of the document
    pdfproducer={Producer}, % producer of the document
    pdfkeywords={Sequential Analysis}{Quickest change-point detection}{Shiryaev--Roberts procedure}{Shiryaev--Roberts--r procedure}{Fredholm integral equations of the second kind}, % list of keywords
    pdfnewwindow=true,      % links in new window
    colorlinks=false,       % false: boxed links; true: colored links
    linkcolor=red,          % color of internal links (change box color with linkbordercolor)
    citecolor=green,        % color of links to bibliography
    filecolor=magenta,      % color of file links
    urlcolor=cyan           % color of external links
]{hyperref}

\usepackage[T1]{fontenc}

\usepackage{titlesec} % Very fancy section title control: Sets sections as SQA requires

\titleformat{\section}{\large\bfseries\uppercase}{\thesection.}{.5em}{}
\titlespacing*{\section}{0pt}{*3}{*2}
\titleformat{\subsection}{\normalfont\bfseries}{\thesubsection.}{.5em}{}
\titlespacing*{\subsection}{0pt}{*3}{*2}
\titleformat{\subsubsection}{\normalfont\bfseries}{\thesubsubsection.}{.5em}{}
\titlespacing*{\subsubsection}{0pt}{*3}{*2}

\numberwithin{equation}{section}

\usepackage[sort&compress]{natbib}

\renewcommand{\Pr}{\mathbb{P}} % probability
\DeclareMathOperator{\EV}{\mathbb{E}} % expected value
\DeclareMathOperator{\Var}{Var}

\DeclareMathOperator{\LR}{\Lambda}

\DeclareMathOperator{\ARL}{ARL}
\DeclareMathOperator{\PFA}{PFA}

\newcommand{\T}{T}

\renewcommand{\le}{\leqslant} % AMS le ge
\renewcommand{\ge}{\geqslant}

\newcommand{\abs}[1]{\left\vert#1\right\vert}

\DeclareMathOperator{\One}{\mathchoice{\rm 1\mskip-4.2mu l}{\rm 1\mskip-4.2mu l}{\rm 1\mskip-4.6mu l}{\rm 1\mskip-5.2mu l}}
\newcommand{\indicator}[1]{{\One_{\left\{#1\right\}}}}

\theoremstyle{plain} %% produces italic text
\newtheorem{theorem}{Theorem}[section]
\newtheorem{lemma}{Lemma}[section]

\theoremstyle{plain}
\newtheorem{remark}{Remark}[section]

\graphicspath{{./gfx/}}

\begin{document}

\title{\textbf{\Large An Accurate Method for Determining the Pre-Change Run-Length Distribution of the Generalized Shiryaev--Roberts Detection Procedure}}

\date{}
\author{}
\maketitle

%%%%%%%%% Authors, affiliations %%%%%%%%%%%%%%%%%%%%%%%%%%
\begin{center}
\null\vskip-2cm\author{
\textbf{\large Aleksey\ S.\ Polunchenko}\\
Department of Mathematical Sciences, State University of New York at Binghamton,\\Binghamton, New York, USA
\vskip0.2cm
\textbf{\large Grigory\ Sokolov}\\
Department of Mathematics, University of Southern California, Los Angeles, California, USA
\vskip0.2cm
\textbf{\large Wenyu\ Du}\\
Department of Mathematical Sciences, State University of New York at Binghamton,\\Binghamton, New York, USA
}
\end{center}
%
%-------------------------------------------------------------------------------------------------%
%
% At the bottom of the first page, provide the e-mail, telephone and fax
% info, and complete address (including the street name, P. O. Box number,
% etc. as needed) for the corresponding author, formatted as follows:
%
% Address correspondence to D. H. Author, Department of Statistics, CLAS
% Building U-4120, University of Connecticut, 215 Glenbrook Road, Storrs, CT % 06269-4120, USA; Fax: 860-486-4113; E-mail: nitis.mukhopadhyay@uconn.edu
%
\symbolfootnote[0]{\normalsize\hspace{-0.6cm}Address correspondence to A.\ S.\ Polunchenko, Department of Mathematical Sciences, State University of New York (SUNY) at Binghamton, Binghamton, NY 13902--6000, USA; Tel: +1 (607) 777-6906; Fax: +1 (607) 777-2450; E-mail:~\href{mailto:aleksey@binghamton.edu}{aleksey@binghamton.edu}.}\\
%
%-------------------------------------------------------------------------------------------------%
%
{\small\noindent\textbf{Abstract:} Change-of-measure is a powerful technique in wide use across statistics, probability and analysis. Particularly known as Wald's likelihood ratio identity, the technique enabled the proof of a number of exact and asymptotic optimality results pertaining to the problem of quickest change-point detection. Within the latter problem's context we apply the technique to develop a numerical method to compute the Generalized Shiryaev--Roberts (GSR) detection procedure's pre-change Run-Length distribution. Specifically, the method is based on the integral-equations approach and uses the collocation framework with the basis functions chosen so as to exploit a certain change-of-measure identity and a specific martingale property of the GSR procedure's detection statistic. As a result, the method's accuracy and robustness improve substantially, even though the method's theoretical rate of convergence is shown to be merely quadratic. A tight upper bound on the method's error is supplied as well. The method is not restricted to a particular data distribution or to a specific value of the GSR detection statistic's ``headstart''. To conclude, we offer a case study to demonstrate the proposed method at work, drawing particular attention to the method's accuracy and its robustness with respect to three factors:\begin{inparaenum}[\itshape a)]\item partition size (rough vs. fine), \item change magnitude (faint vs. contrast), and \item Average Run Length (ARL) to false alarm level (low vs. high)\end{inparaenum}. Specifically, assuming independent standard Gaussian observations undergoing a surge in the mean, we employ the method to study the GSR procedure's Run-Length's pre-change distribution, its average (i.e., the usual ARL to false alarm) and standard deviation. As expected from the theoretical analysis, the method's high accuracy and robustness with respect to the foregoing three factors are confirmed experimentally. We also comment on extending the method to handle other performance measures and other procedures.
}
\\ \\
%-------------------------------------------------------------------------------------------------%
% KEYWORDS
%
% Keywords are to be listed alphabetically (with the first letter
% capitalized), preferably chosen from the text and not from the title of the % paper itself; the importance of the words used in a title is already
% obvious to readers. The keywords are to be separated by semicolons (;).
%
\\ \\
%-------------------------------------------------------------------------------------------------%
% SUBJECT CLASSIFICATIONS
%
% MSC2010, see http://www.ams.org/msc/
%
%  Statistics -> Sequential methods
%   62L10 - Sequential analysis
%   62L15 - Optimal stopping
%
%  Statistics -> Applications
%   62P30 - Applications in engineering and industry
%
%  Numerical analysis -> Integral equations
%

%+-----------------------------------------------------------------------------------------------+%
\section{Introduction}
\label{sec:intro}

Sequential (quickest) change-point detection is concerned with the design and analysis of statistical procedures for rapid ``on-the-go'' detection of possible spontaneous changes in the characteristics of a ``live'' monitored (random) process. Specifically, the process is assumed to be monitored continuously through sequentially made observations, and if there is an indication of a possible change in the process' characteristics, one is to detect it as soon as possible, subject to a tolerable level of the risk of false detection. This type of statistical process control is used, e.g., in industrial quality and process control (see, e.g.,~\citealp{Ryan:Book2011,Montgomery:Book2012,Kenett+Zacks:Book1998}), biostatistics, economics, seismology (see, e.g.,~\citealp[Section~11.1.2]{Basseville+Nikiforov:Book93}), forensics, navigation (see, e.g.,~\citealp[Section~11.1.1]{Basseville+Nikiforov:Book93}), cybersecurity (see, e.g.,~\citealp{Tartakovsky+etal:SM2006-discussion,Tartakovsky+etal:JSM2005,Tartakovsy+etal:IEEE-JSTSP2013,Polunchenko+etal:SA2012}), and communication systems. A sequential change-point detection procedure, a rule whereby one is to stop and declare that (apparently) a change is in effect, is defined as a stopping time, $\T$, that is adapted to the observed data, $\{X_n\}_{n\ge1}$.

It will be assumed in this work that the observations, $\{X_n\}_{n\ge1}$, are independent throughout the entire period of surveillance, and such that $X_1,\ldots,X_{\nu}$ are distributed according to a common density $f(x)$, while $X_{\nu+1},X_{\nu+2},\ldots$ each follow a common density $g(x)$; both $f(x)$ and $g(x)$ are fully specified, and $g(x)\not\equiv f(x)$. The serial number of the last $f(x)$-distributed observation, i.e., $\nu\ge0$, is the change-point, and we will regard it as unknown (but not random); the notation $\nu=0$ ($\nu=\infty$) is to be understood as the case when the density of $X_n$ is $g(x)$ ($f(x)$) for all $n\ge1$. This is schematically illustrated in Figure~\ref{fig:basic-iid-change-point}.
\begin{figure}[!htb]
    \centering
    \includegraphics[width=0.8\textwidth]{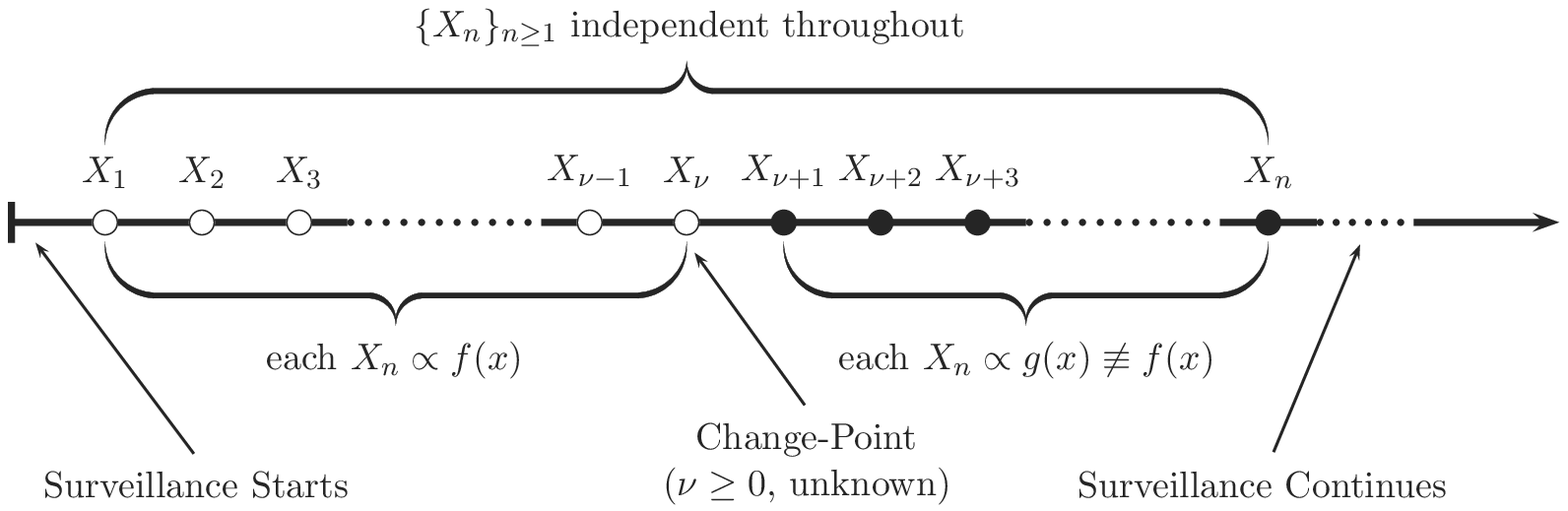}
    \caption{Observations model.}
    \label{fig:basic-iid-change-point}
\end{figure}

At this point, by appropriately choosing the detection delay loss function (which is to be minimized) the described observations model can be used to formulate the corresponding change-point detection problem under at least three different criteria: minimax, generalized (or quasi-) Bayesian and multi-cyclic. For a recent survey, see, e.g.,~\cite{Tartakovsky+Moustakides:SA10},~\cite{Polunchenko+Tartakovsky:MCAP2012} or~\cite{Polunchenko+etal:JSM2013}. However, the focus of this work is on the pre-change regime only. More specifically, the bulk of this paper is devoted to determining the pre-change performance of two related detection procedures: the original Shiryaev--Roberts (SR) procedure (due to the independent work of~\citealp{Shiryaev:SMD61,Shiryaev:TPA63} and that of~\citealp{Roberts:T66}) and its generalization -- the Shiryaev--Roberts--$r$ (SR--$r$) procedure introduced by~\cite{Moustakides+etal:SS11}. As the former procedure is a special case of the latter, we will refer to the SR--$r$ procedure as the Generalized SR (GSR) procedure, in analogy to the terminology used by~\cite{Tartakovsky+etal:TPA2012}.

It is a fact that there isn't much literature on numerical methodology to evaluate the GSR procedure's performance, even though the procedure is nearly minimax optimal in the sense of~\cite{Pollak+Siegmund:AS1975} and~\cite{Pollak:AS85}, and exactly optimal in Shiryaev's~\citeyearpar{Shiryaev:SMD61,Shiryaev:TPA63} multi-cyclic setup. See~\cite{Tartakovsky+etal:TPA2012,Pollak+Tartakovsky:SS09} and~\cite{Polunchenko+Tartakovsky:AS10} for the relevant results. A fair amount of research has been done on the Girschick--Rubin~\citeyearpar{Girschick+Rubin:AMS52} procedure, a ``precursor'' of the original SR procedure; see, e.g.,~\cite{Knoth:FSQC2006,Yashchin:ISR1993}. The original SR procedure was studied, e.g., by~\cite[Section~11.6]{Kenett+Zacks:Book1998}. Also, a comparative analysis of the original SR procedure against various ``mainstream'' charts was carried out, e.g., by~\cite{Mahmoud+etal:JAS2008,Mevorach+Pollak:AJMMS91,Moustakides+etal:CommStat09}. Further headway was recently made by~\cite{Moustakides+etal:CommStat09,Moustakides+etal:SS11,Tartakovsky+etal:IWSM2009}. They offered an efficient numerical method to compute a wide range of operating characteristics of not only the GSR procedure, but of a much larger family of procedures, spanning in particular the popular Cumulative Sum (CUSUM) chart due to~\cite{Page:B54} and the Exponentially Weighted Moving Average (EWMA) scheme proposed by~\cite{Roberts:T59}. However, the question of the method's accuracy was only partially answered, with no error bounds or convergence rates supplied. This agrees with the overall trend in the literature concerning the computational aspect of change-point detection: to deal with the accuracy question in an {\it ad hoc} manner, if at all. This work's aim is twofold: First, we build on to the previous work of~\cite{Moustakides+etal:CommStat09,Moustakides+etal:SS11,Tartakovsky+etal:IWSM2009} and develop a more accurate and robust numerical method to study the GSR procedure's Run-Length distribution. We target the pre-change regime only, and focus on the entire pre-change distribution of the corresponding stopping time and on its first moment (i.e., the usual Average Run Length to false alarm) and standard deviation. The method is based on the integral-equations approach and uses the collocation framework as the principal idea. The basis functions are selected to be linear ``hat'' functions so as to make use of a certain change-of-measure identity and also exploit a specific martingale property of the GSR procedure's detection statistic. Second, the design and analysis of the method are complete in that the question of the method's accuracy is properly addressed, providing a tight error bound and convergence rate. Due to the use of a change-of-measure identity in combination with a certain martingale property of the GSR procedure the method's accuracy is rather high even if the partition is rough, be it because the number of partition nodes is too low or because the partitioned interval is too wide; the latter is an issue when the procedure's detection threshold is high, i.e., in the asymptotic case. Furthermore, the method's accuracy is also robust with respect to the magnitude of the change (faint, moderate or contrast).

The rest of the paper is organized thus. We start with formally stating the problem and describing the GSR procedure in Section~\ref{sec:GSR+properties}. The numerical method and its accuracy analysis are presented in Section~\ref{sec:GSR-RL-evaluation}. In Section~\ref{sec:case-study} we put the numerical method to the test by considering a case study. The study is based on the problem of detecting a shift in the mean of a sequence of independent Gaussian random variables. For this scenario, the performance of the original SR procedure and its extensions has already been quantified and compared against that of the CUSUM chart; see, e.g.,~\cite{Moustakides+etal:CommStat09,Tartakovsky+etal:IWSM2009,Mahmoud+etal:JAS2008,Knoth:FSQC2006,Mevorach+Pollak:AJMMS91}. We tailor the case study to assessing the accuracy of the numerical method, and to computing the GSR procedure's actual Run-Length distribution and its characteristics. To the best of our knowledge, no such work has previously been done. In Section~\ref{sec:further-discussion} we comment on extending the proposed method to other performance measures as well as to other detection procedures. Finally, in Section~\ref{sec:conclusion} we draw conclusions.

%+-----------------------------------------------------------------------------------------------+%
\section{The Generalized Shiryaev--Roberts procedure}
\label{sec:GSR+properties}

We first set down some notation. Let $\T$ be a generic stopping time. Let $\Pr_\nu$ and $\EV_\nu$ be, respectively, the probability measure and the corresponding expectation given a known $0\le\nu\le\infty$. Particularly, $\Pr_\infty$ and $\EV_\infty$ denote, respectively, the probability measure and the corresponding expectation assuming the observations' distribution never changed (i.e., $\nu=\infty$). Likewise, $\Pr_0$ and $\EV_0$ are assuming the distribution had changed before the first observation was made (i.e., $\nu=0$).

One of the objects of interest in this work is the Average Run Length (ARL) to false alarm, defined as $\ARL(\T)\triangleq\EV_\infty[\T]$. This is the standard minimax measure of the false alarm risk introduced by Lorden's~\citeyearpar{Lorden:AMS71}. However, this is just the first moment of the $\Pr_\infty$-distribution of the corresponding stopping time. As such, it may not be an exhaustive measure of the risk of sounding a false alarm; see, e.g.,~\cite{Lai:IEEE-IT1998} and~\cite{Tartakovsky:IEEE-CDC05}. Therefore, we will also be interested in the entire $\Pr_\infty$-distribution of the corresponding stopping time, and in particular, in its standard deviation, which we will denote as $\sqrt{\Var_\infty[\T]}\triangleq\sqrt{\EV_\infty[\T^2]-(\EV_\infty[\T])^2}$.

To introduce the GSR procedure we first construct the corresponding likelihood ratio (LR). Let $\mathcal{H}_k\colon\nu=k$, where $0\le k<\infty$ and $\mathcal{H}_{\infty}\colon\nu=\infty$ be, respectively, the hypothesis that the change takes place at time moment $\nu=k$, where $0\le k<\infty$, and the alternative hypothesis that no change ever occurs. The joint densities of the sample $\boldsymbol{X}_{1:n}\triangleq(X_1,\ldots,X_n)$, $n\ge1$, under each of these hypotheses are given by
\begin{align*}
p(\boldsymbol{X}_{1:n}|\mathcal{H}_{\infty})
&=
\prod_{j=1}^n f(X_j)
\;\text{and}\;
p(\boldsymbol{X}_{1:n}|\mathcal{H}_k)
=
\prod_{j=1}^k f(X_j)\prod_{j=k+1}^n
g(X_j)\;\text{for $k<n$},
\end{align*}
assuming $p(\boldsymbol{X}_{1:n}|\mathcal{H}_{\infty})=p(\boldsymbol{X}_{1:n}|\mathcal{H}_k)$ for $k\ge n$. Hence, the corresponding LR is
\begin{align*}
\LR_{1:n,\nu=k}
&\triangleq
\frac{p(\boldsymbol{X}_{1:n}|\mathcal{H}_k)}{p(\boldsymbol{X}_{1:n}|\mathcal{H}_{\infty})}
=\prod_{j=k+1}^n\LR_j,\;\text{for $k<n$},
\end{align*}
where $\LR_n\triangleq g(X_n)/f(X_n)$ is the LR for the $n$-th observation, $X_n$.

We now make an observation that will play an important role in the sequel. Let $P_d^{\LR}(t)\triangleq\Pr_d(\LR_1\le t)$, $t\ge0$, $d=\{0,\infty\}$, denote the cumulative distribution function (cdf) of the LR under measure $\Pr_d$, $d=\{0,\infty\}$, respectively. Since the LR is the Radon--Nikod\'{y}m derivative of measure $\Pr_0$ with respect to measure $\Pr_\infty$ (the two measures are assumed to be mutually absolutely continuous), one can conclude that $dP_0^{\LR}(t)=tdP_\infty^{\LR}(t)$. To be more specific, this can be seen from the following argument:
\begin{align*}
dP_0^{\LR}(t)
&\triangleq d\Pr_0(\LR_1\le t)\\
&=
d\Pr_0(X_1\le \LR_1^{-1}(t))\\
&=
\LR(\LR^{-1}(t))\,d\Pr_\infty(X_1\le \LR_1^{-1}(t))\\
&=
td\Pr_\infty(\LR_1\le t)\\
&=
tdP_\infty^{\LR}(t).
\end{align*}
We will rely on this change-of-measure identity later in Section~\ref{sec:GSR-RL-evaluation} to design our numerical method: the identity and a martingale property of the GSR procedure discussed below will be key to improve the method's accuracy and stability.

\begin{remark}
The above change-of-measure identity can also be derived from Wald's~\citeyearpar{Wald:Book47} likelihood ratio identity; see also, e.g.,~\cite[p.~13]{Siegmund:Book85},~\cite[p.~4]{Woodroofe:Book82}, or~\cite{Lai:SA2004}. Alternatively, since by design $\LR_n\in[0,\infty)$ with probability 1 for all $n\ge1$, and $\EV_\infty[\LR_n]=1$ for all $n\ge1$, one can regard $dP_0^{\LR}(t)=tdP_\infty^{\LR}(t)$ as a size-bias probability-measure transformation; see, e.g.,~\cite{Arratia+Goldstein:arXiv2010}.
\end{remark}

Formally, the SR procedure is defined as the stopping time $\mathcal{S}_A\triangleq\inf\big\{n\ge1\colon R_n\ge A\big\}$, where $A>0$ is a detection threshold that controls the false alarm risk, and
\begin{align}\label{eq:Rn-SR-def}
R_n
&=
\sum_{k=1}^n\LR_{1:n,\nu=k}=\sum_{k=1}^n \prod_{i=k}^n\LR_i,\; n\ge1,
\end{align}
is the SR detection statistic; note the recursion $R_{n}=(1+R_{n-1})\LR_n$ for $n=1,2,\ldots$, where $R_0=0$.

We remark that the detection threshold, $A>0$, can be large, making it unlikely for the detection statistic, $\{R_n\}_{n\ge1}$, to hit (much less to cross) it. As a result, $\{n\ge1\colon R_n\ge A\big\}$ is asymptotically (as $A\to\infty$) an empty set, and therefore the stopping rule, $\mathcal{S}_A$, will never terminate. To account for this, we will assume here and in every definition of a detection procedure to follow that $\inf\{\varnothing\}=\infty$.

Observe that $\{R_n-n\}_{n\ge0}$ is a zero-mean $\Pr_\infty$-martingale, i.e., $\EV_\infty[R_n-n]=0$. From this and the Optional stopping theorem (see, e.g.,~\citealp[Subsection~2.3.2]{Poor+Hadjiliadis:Book08}) one can conclude that $\EV_\infty[R_{\mathcal{S}_A}-\mathcal{S}_A]=0$ so that $\ARL(\mathcal{S}_A)\triangleq\EV_\infty[\mathcal{S}_A]=\EV_\infty[R_{\mathcal{S}_A}]\ge A$. Hence, it is sufficient to set $A=\gamma$ to ensure $\ARL(\mathcal{S}_A)\ge\gamma$ for any desired $\gamma>1$. More precisely,~\cite{Pollak:AS87} established that
\begin{align*}
\ARL(\mathcal{S}_A)
&=
\dfrac{A}{\xi}[1+o(1)],\;\text{as}\;\gamma\to\infty,
\end{align*}
where $\xi\in(0,1)$ is the limiting exponential overshoot. It is model-dependent constant and can be computed using nonlinear renewal theory; see, e.g.,~\cite{Siegmund:Book85} or~\cite{Woodroofe:Book82}. For practical purposes the approximation $\ARL(\mathcal{S}_A)\approx A/\xi$ is known to be extremely accurate, even if the corresponding ARL is in the tens.

We now describe the SR--$r$ procedure. It was introduced by~\cite{Moustakides+etal:SS11} who regarded starting off the original SR procedure at a fixed (but specially designed) $R_0^r=r$, $r\ge0$. This is similar to the idea proposed earlier by~\cite{Lucas+Crosier:T1982} for Page's~\citeyearpar{Page:B54} CUSUM chart. However, as shown by~\cite{Moustakides+etal:SS11} and~\cite{Tartakovsky+etal:TPA2012}, giving the SR procedure a headstart is practically putting it on steroids: the gain in performance far exceeds that observed by~\cite{Lucas+Crosier:T1982} for the CUSUM chart.

Specifically, the SR--$r$ procedure is defined as the stopping time
\begin{align}\label{eq:T-SRr-def}
\mathcal{S}_{A}^r
&=
\inf\{n\ge1\colon R_n^r \ge A\} ,\; A >0,
\end{align}
where
\begin{align}\label{eq:statistic-SRr-def}
R_{n+1}^r
&=
(1+R_{n}^r)\LR_{n+1}\;\;\text{for}\;\; n=1,2,\ldots,\;\;\text{where}\;\;R_0^r=r,
\end{align}
and we remark that $\{R_n^r-n-r\}_{n\ge0}$ is a zero-mean $\Pr_\infty$-martingale, i.e., $\EV_\infty[R_n^r-n-r]=0$. As a result, one can establish that
\begin{align*}
\ARL(\mathcal{S}_A^r)
&\approx
\dfrac{A}{\xi}-r,\;\text{as}\;\gamma\to\infty,
\end{align*}
where $\xi\in(0,1)$ is again the limiting exponential overshoot. This approximation is quite accurate as well, even if the ARL is small.

We now make a comment on the foregoing approximation for $\ARL(\mathcal{S}_A^r)$. For a fixed $A>0$ it may seem as though $\ARL(\mathcal{S}_A^r)$ is a linear function of $r$ with a slope of $-1$ and a constant $y$-intercept (proportional to $A$). However, this is true only if the headstart is not too large. For otherwise $\ARL(\mathcal{S}_A^r)$ would have zero or negative values, which is impossible since by definition $\ARL(\mathcal{S}_A^r)\ge1$ for all $r,A>0$. If $r$ is too large, then $\ARL(\mathcal{S}_A^r)$ is close to $1$, and in fact $\lim_{r\to\infty}\ARL(\mathcal{S}_A^r)=1$. This is illustrated in Figure~\ref{fig:ARL_vs_r}.
\begin{figure}[!htb]
    \centering
    \includegraphics[width=0.9\textwidth]{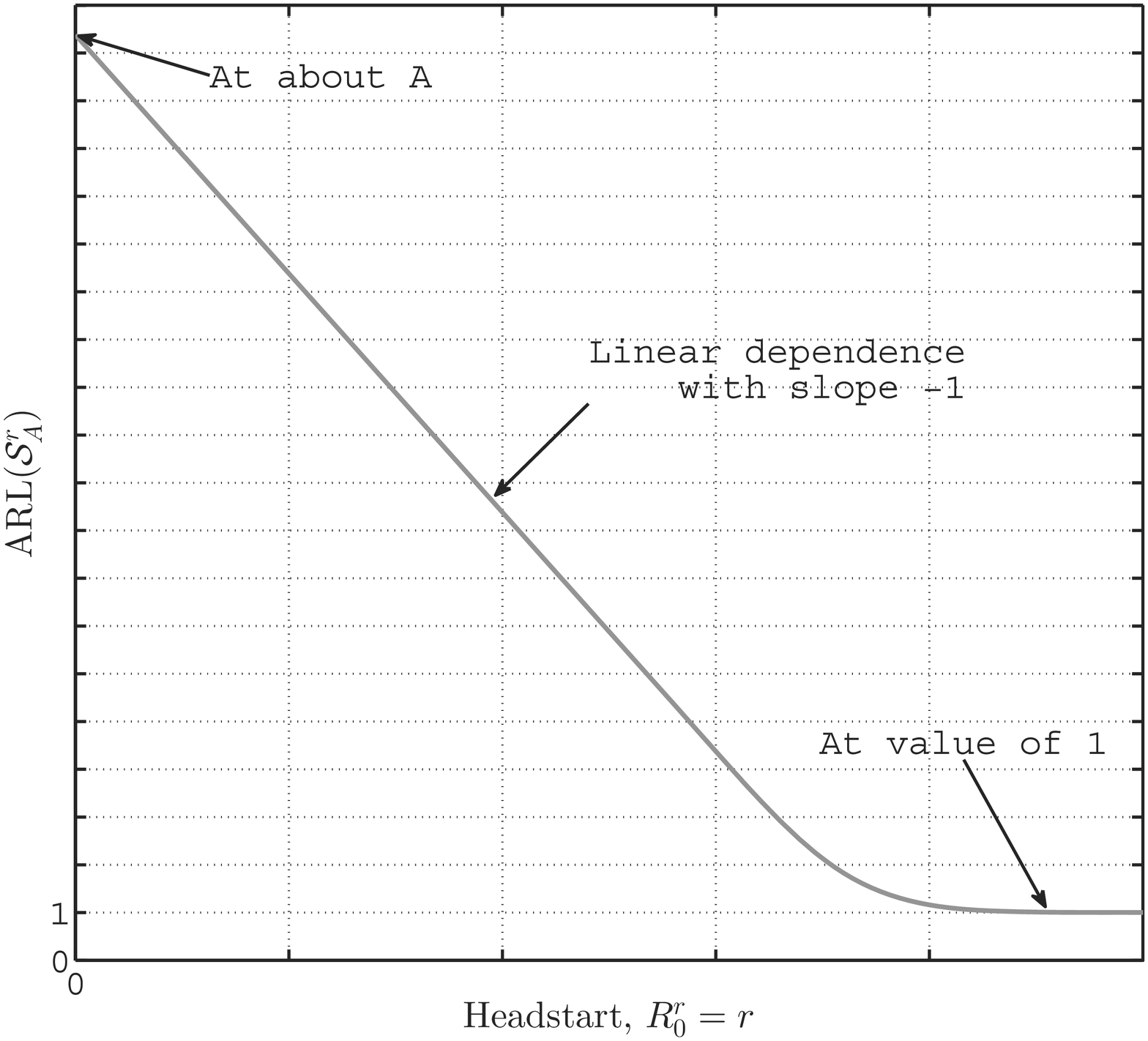}
    \caption{Typical behavior of $\ARL(\mathcal{S}_A^r)$ as a function of $r\ge0$ for a fixed $A>0$.}
    \label{fig:ARL_vs_r}
\end{figure}
Since $\ARL(\mathcal{S}_A^{r=0})\approx A/\xi>A$ and as a function of $r$ it's a line with a slope of $-1$, the value of the headstart at which $\ARL(\mathcal{S}_A^{r=0})$ starts to curl is also about $r_{\mathrm{crit}}\approx A/\xi$. Therefore, roughly speaking for $r\in[0,A/\xi]$, the behavior of $\ARL(\mathcal{S}_A^{r})$ as a function of $r$ is almost linear with a slope of $-1$. This is a direct consequence of the GSR procedure's martingale property mentioned above, and it will be used to improve the accuracy of our numerical method.

We now make a comment on the foregoing approximation for $\ARL(\mathcal{S}_A^r)$. For a fixed $A>0$ it may seem as though $\ARL(\mathcal{S}_A^r)$ is a linear function of $r$ with a slope of $-1$ and a constant $y$-intercept (linearly proportional to $A$). However, this is true only if the headstart is not too large, in which case $\ARL(\mathcal{S}_A^r)$ is close to $1$, and in fact $\lim_{r\to\infty}\ARL(\mathcal{S}_A^r)=1$. This is illustrated in Figure~\ref{fig:ARL_vs_r}.
Since $\ARL(\mathcal{S}_A^{r=0})\approx A/\xi>A$, the value of the headstart at which $\ARL(\mathcal{S}_A^{r=0})$ starts to curl is also about $r_{\mathrm{crit}}\approx A/\xi$. Therefore, roughly speaking, for $r\in[0,A/\xi]$ the behavior of $\ARL(\mathcal{S}_A^{r})$ as a function of $r$ is almost linear. This is a direct consequence of the GSR procedure's martingale property mentioned above, and it will be used to improve the accuracy of our numerical method.

We have now introduced all of the procedures of interest. Next, we present a numerical method to study their operating characteristics.

%+-----------------------------------------------------------------------------------------------+%
\section{Run-Length distribution}
\label{sec:GSR-RL-evaluation}

We develop here a numerical method to study the Run-Length distribution of the GSR procedure; the corresponding stopping time, $\mathcal{S}_A^r$, is given by~\eqref{eq:T-SRr-def} and~\eqref{eq:statistic-SRr-def} above. Specifically, the focus is on the ``in-control'' distribution and its characteristics. The distribution itself is given by $\Pr_\infty(\mathcal{S}_A^r=k)$, $k\ge0$, i.e., the probability mass function (pmf) of the GSR stopping time; we will also consider $\Pr_\infty(\mathcal{S}_A^r>k)$, $k\ge0$, i.e., the corresponding ``survival function''. The distribution's characteristics of interest are\begin{inparaenum}[\itshape a)]\item the first moment (under the $\Pr_\infty$ measure), i.e., the ARL to false alarm defined as $\ARL(\mathcal{S}_A^r)\triangleq\EV_\infty[\mathcal{S}_A^r]$ and \item the distribution's higher moments (all under the $\Pr_\infty$ measure)\end{inparaenum}. The proposed numerical method is a build-up over one previously proposed and applied by~\cite{Tartakovsky+etal:IWSM2009,Moustakides+etal:CommStat09,Moustakides+etal:SS11}.

%+-----------------------------------------------------------------------------------------------+%
\subsection{Integral Equations Framework}
\label{ssec:integral-equations-framework}

We begin with notation and assumptions. First recall $\LR_n\triangleq g(X_n)/f(X_n)$, i.e., the ``instantaneous'' LR for the $n$-th data point, $X_n$, and note that $\{\LR_n\}_{n\ge1}$ are iid under both measures $\Pr_d$, $d=\{0,\infty\}$. For simplicity, $\LR_1$ will be assumed absolutely continuous, although at an additional effort the case of non-arithmetic $\LR_1$ can be handled as well. Let $P_d^{\LR}(t)\triangleq\Pr_d(\LR_1\le t)$, $t\ge0$, be the cdf of the LR under the measure $\Pr_d$, $d=\{0,\infty\}$. Also, denote
\begin{align}\label{eq:K-def}
K_d(x,y)
&\triangleq
\frac{\partial}{\partial y}\Pr_d(R_{n+1}^x\le y|R_n^x=x)
=
\frac{\partial}{\partial y}P_d^{\LR}\left(\frac{y}{1+x}\right),
\; d=\{0,\infty\},
\end{align}
the {\em transition probability density kernel} for the (stationary) Markov process $\{R_n^x\}_{n\ge0}$.

We now note that from~\eqref{eq:K-def} and the change-of-measure identity $d{P}_0^{\LR}(t)=t d{P}_{\infty}^{\LR}(t)$, $t\ge0$ established earlier, one can readily deduce that $(1+x){K}_0(x,y)=y{K}_{\infty}(x,y)$. This can be used, e.g., as a ``shortcut'' in deriving the formula for ${K}_0(x,y)$ from that for ${K}_{\infty}(x,y)$, or the other way around -- whichever one of the two is found first. More importantly, this connection between ${K}_0(x,y)$ and ${K}_{\infty}(x,y)$ will allow our numerical method to have greater accuracy and robustness.

We now state the first equation of interest. Let $R_0^x=x\ge-1$ be fixed. For notational brevity, from now on let $\ell(x,A)\triangleq\ARL(\mathcal{S}_A^x)\triangleq\EV_{\infty}[\mathcal{S}_A^x]$; we reiterate that this expectation is conditional on $R_0^x=x$. Using the fact that $\{R_n^x\}_{n\ge0}$ is Markovian, it can be shown that $\ell(x,A)$ satisfies the equation
\begin{align}\label{eq:ARL-int-eqn}
\ell(x,A)
&=
1+\int_0^AK_{\infty}(x,y)\,\ell(y,A)\,dy;
\end{align}
cf.~\cite{Moustakides+etal:SS11}. It is also worth noting here that so long as $R_0^x=x\ge-1$, one can infer from~\eqref{eq:Rn-SR-def} that $R_n^x\ge0$ with probability 1 for all $n\ge1$ under either measure $\Pr_d$, $d=\{0,\infty\}$. This is why the lower limit of integration in the integral in the right-hand side of~\eqref{eq:ARL-int-eqn} is 0. Furthermore, due to the GSR procedure's martingale property discussed earlier, we have $\ell(x,A)\approx A/\xi-x$ at least for $x\le A/\xi$. Any good numerical method to solve the equation for $\ell(x,A)$ is to factor in the expected shape of the sought function.

Next, introduce $\rho_k(x,A)\triangleq\Pr_\infty(\mathcal{S}_A^x>k)$ for a fixed $k\ge0$; note that $\rho_0(x,A)\equiv1$ for all $x,A\in\mathbb{R}$, and in general $0\le\rho_k(x,A)\le1$ for all $x,A\in\mathbb{R}$ and $k\ge0$. We will occasionally refer to $\rho_k(x,A)$ as the ``survival function''. Again, since $\{R_n^x\}_{n\ge0}$ is Markovian, one can establish the recurrence
\begin{align}\label{eq:recursion2}
\rho_{k+1}(x,A)
&=
\int_0^A K_{\infty}(x,y)\,\rho_k(y,A)\,dy,
\end{align}
wherewith one can generate the functional sequence $\{\rho_k(x,A)\}_{k\ge0}$; cf.~\cite{Moustakides+etal:SS11}. Note that this recurrence is repetitive application of the linear integral operator
\begin{align*}
\mathcal{K}_{\infty}\circ u
&=
\int_0^A K_{\infty}(x,y)\,u(y)\,dy,
\end{align*}
where $u(x)$ is assumed to be sufficiently smooth inside the interval $[0,A]$. Temporarily deferring formal discussion of this operator's properties, note that using this operator notation, recurrence~\eqref{eq:recursion2} can be rewritten as $\rho_{k+1}=\mathcal{K}_{\infty}\circ \rho_{k}$, $k\ge0$, or equivalently, as $\rho_{k}=\mathcal{K}_{\infty}^{k}\circ \rho_{0}$, $k\ge0$, where
\begin{align*}
\mathcal{K}_{\infty}^{k}\circ u
&\triangleq
\underbrace{\mathcal{K}_{\infty}\circ\cdots\circ\mathcal{K}_{\infty}}_{\text{$k$ times}}\circ\,u\;\text{for}\;k\ge1,
\end{align*}
and $\mathcal{K}_{\infty}^{0}$ is the identity operator from now on denoted as $\mathbb{I}$, i.e., $\mathcal{K}_{\infty}^{0}\circ u=\mathbb{I}\circ u=u$. Similarly, in the operator form, equation~\eqref{eq:ARL-int-eqn} can be rewritten as $\ell=1+\mathcal{K}_\infty\circ\ell$.

We note that the sequence $\{\rho_k(x,A)\}_{k\ge0}$ can be used to derive many characteristics of the GSR stopping time. For example, since
\begin{align}\label{eq:ARL-Neumann-series}
\ell
&=
\sum_{k\ge0}\rho_k=\sum_{k\ge0}\mathcal{K}_\infty^k\circ\rho_0
=
\left(\,\sum_{k\ge0}\mathcal{K}_\infty^k\right)\circ\rho_0
=
(\,\mathbb{I}-\mathcal{K}_\infty)^{-1}\circ\rho_0,
\end{align}
one obtains equation~\eqref{eq:ARL-int-eqn} recalling that $\rho_0(x,A)=1$ for all $x,A\in\mathbb{R}$. The implicit use of the geometric series convergence theorem to reduce the last infinite sum to the operator $(\,\mathbb{I}-\mathcal{K}_\infty)^{-1}$ is justified by the fact that the spectral radius of the operator $\mathcal{K}_\infty$ is strictly less than 1. This is shown, e.g.,~\cite{Moustakides+etal:SS11}.

More generally, through the sequence $\{\rho_k(x,A)\}_{k\ge0}$, one can derive an equation for the characteristic function of the GSR stopping time, and then use it to obtain an equation for any moment of the GSR stopping time (under the $\Pr_\infty$ measure). Specifically, denote $\Psi(s,x,A)\triangleq\EV_\infty[e^{\imath s\mathcal{S}_A^x}]$, where hereafter $\imath=\sqrt{-1}$. Similarly to~\eqref{eq:ARL-Neumann-series}, it can be shown that $\Psi(s,x,A)$ solves the equation
\begin{align}\label{eq:GSR-char-fcn-eqn}
\Psi(s,x,A)
&=
e^{\imath s}[1-\rho_1(x,A)]
+e^{\imath s}\int_0^A K_\infty(x,y)\,\Psi(s,y,A)\,dy,
\end{align}
where $\rho_1(x,A)$ is as in~\eqref{eq:recursion2} above, and $\Psi(0,x,A)\equiv1$ for all $x,A\in\mathbb{R}$.

To illustrate the power of equation~\eqref{eq:GSR-char-fcn-eqn}, consider the second moment of the GSR stopping time (under the $\Pr_\infty$ measure). Let $\mu_2(x,A)\triangleq\EV_\infty[(\mathcal{S}_A^x)^2]$, and note that $\mu_2(x,A)=-\partial^2 \Psi(s,x,A)/\partial s^2$ evaluated at $s=0$. Hence, from differentiating equation~\eqref{eq:GSR-char-fcn-eqn} twice with respect to $s$, and then setting $s=0$, we obtain
\begin{align}\label{eq:mu2-int-eqn}
\mu_2(x,A)
&=
2\ell(x,A)-1
+\int_0^A K_\infty(x,y)\,\mu_2(y,A)\,dy,
\end{align}
where $\ell(x,A)$ is as in~\eqref{eq:ARL-int-eqn}. Respective equations for the third and higher moments of the GSR stopping time (under the $\Pr_\infty$ measure), i.e., $\mu_n(x,A)\triangleq\EV_\infty[(\mathcal{S}_A^x)^n]$, $n\ge3$, can be obtained in an analogous manner by evaluating $\mu_n(x,A)=(-\imath)^n\partial^n \Psi(s,x,A)/\partial s^n$ at $s=0$.

We remark parenthetically that equation~\eqref{eq:GSR-char-fcn-eqn} resembles that for the characteristic function of the geometric distribution (over the set $\{1,2,\ldots\}$). One may therefore deduce that the $\Pr_\infty$-distribution of the GSR stopping time, $\mathcal{S}_A^x$, is geometric. This is, in fact, true, though generally only in the limit, as $A\to\infty$, and is a special case of the more general result of~\cite{Pollak+Tartakovsky:TVP09,Tartakovsky+etal:IWAP08}. Under certain mild conditions, they showed that the first exit time of a nonnegative recurrent monotone Markov process from the strip $[0,B]$, $B>0$, is asymptotically (as $B\to\infty$) geometrically distributed; the ``probability of success'' is the reciprocal of the first exit time's expectation.

The series $\{\rho_k(x,A)\}_{k\ge0}$ can also be used to obtain the $\Pr_\infty$-pmf of the GSR's stopping time. To this end, it suffices to note that $\Pr_\infty(\mathcal{S}_A^x=k)=\Pr_\infty(\mathcal{S}_A^x>k-1)-\Pr_\infty(\mathcal{S}_A^x>k)=\rho_{k-1}(x,A)-\rho_{k}(x,A)$, $k\ge1$, whence
\begin{align*}
\Pr_\infty(\mathcal{S}_A^x=k)
&=
(\mathbb{I}-\mathcal{K}_\infty)\circ\rho_{k-1}
=
(\mathbb{I}-\mathcal{K}_\infty)\circ\mathcal{K}_\infty^{k-1}\circ\rho_0
=
\mathcal{K}_\infty^{k-1}\circ(\mathbb{I}-\mathcal{K}_\infty)\circ\rho_0,\;k\ge1,
\end{align*}
since the operators $\mathcal{K}_\infty$ and $\mathbb{I}-\mathcal{K}_\infty$ commute with one another.

We have now obtained a complete set of integral equations and relations to compute any of the desired performance characteristics of the GSR stopping time. The question to be considered next is that of computing these characteristics in practice.

%+-----------------------------------------------------------------------------------------------+%
\subsection{The Numerical Method and Its Accuracy}
\label{ssec:numerical-method+accuracy}

We now turn attention to the question of solving the equations presented in the preceding subsection. To this end, observe first that equations~\eqref{eq:ARL-int-eqn} and~\eqref{eq:mu2-int-eqn} are both renewal-type equations of the general form:
\begin{align*}
u(x)
&=
\upsilon(x)+\int_0^A K_\infty(x,y)\,u(y)\,dy,
\end{align*}
where $\upsilon(x)$ is a given (known) function, $K_\infty(x,y)$ is as in~\eqref{eq:K-def}, and $u(x)$ is the unknown to be found; note that while $u(x)$ does depend on the upper limit of integration, $A$, for notational simplicity we will no longer emphasize that, and use $u(x)$ instead of $u(x,A)$.

To see that the above equation is an ``umbrella'' equation for the equations of interest, observe that, e.g., to obtain equation~\eqref{eq:ARL-int-eqn} on the ARL to false alarm, it suffices to set $\upsilon(x)\equiv 1$ for any $x\in\mathbb{R}$. Similarly, choosing $\upsilon(x)=2\ell(x,A)-1$, where $\ell(x,A)$ is as in~\eqref{eq:ARL-int-eqn}, will yield equation~\eqref{eq:mu2-int-eqn}, which governs the second $\Pr_\infty$-moment of the GSR stopping time. Thus, solving the above equation is the same as solving~\eqref{eq:ARL-int-eqn} and~\eqref{eq:mu2-int-eqn}. The problem, however, is that the general equation is a Fredholm integral equation of the second kind, and such equations rarely allow for an analytical solution. Hence, a numerical approach is in order.

We first set the underlying space for the problem. Let $\mathcal{X}=\mathbb{C}[0,A]$ be the space of continuous functions over the interval $[0,A]$. Equip $\mathcal{X}$ with the usual uniform norm. We will assume that ${P}_\infty^{\LR}(t)$ and the unknown function $u(x)$ are both differentiable as far as necessary. Under these assumptions $\mathcal{K}_\infty$ is a bounded linear operator from $\mathcal{X}$ into $\mathcal{X}$, equipped with the usual $\mathrm{L}_\infty$ norm, $\|u\|_\infty\triangleq\max_{x\in[0,A]}|u(x)|$. Specifically,
\begin{align*}
\|\mathcal{K}_\infty\|_\infty
&\triangleq
\sup_{x\in[0,A]}\int_0^A\abs{K_\infty(x,y)}dy < 1;
\end{align*}
cf.~\cite{Moustakides+etal:SS11}. An important implication of this result is that one can apply the Fredholm alternative (see, e.g.,~\citealp[Theorem~2.8.10]{Atkinson+Han:Book09}) and conclude that all of the equations of interest do have a solution and it is unique. The same conclusion can be reached from the Banach fixed-point theorem, proved by Banach in his Ph.D. thesis in 1920 (published in 1922).

We now define a projection space within which our collocation solution, $u(x)$, will then be sought. Under this method $u(x)$ is approximated as
\begin{align*}
u_N(x)
&=
\sum_{j=1}^N u_{j,N}\,\phi_j(x),
\end{align*}
where $\{u_{j,N}\}_{1\le j\le N}$ are constant coefficients to be determined, and $\{\phi_j(x)\}_{1\le j\le N}$ are suitably chosen (known) basis functions.

We now discuss a strategy for determining the coefficients $\{u_{j,N}\}_{1\le j\le N}$. For any fixed set of these coefficients, substitution of $u_N(x)$ into the equation will give a residual $r_N=u_N-\mathcal{K}_{\infty}\circ u_N-\upsilon$. Unless the true solution $u(x)$ itself is a linear combination of the basis functions $\{\phi_j(x)\}_{1\le j\le N}$, no choice of the coefficients $\{u_{j,N}\}_{1\le j\le N}$ will make the residual identically zero uniformly for all $x\in[0,A]$. However, by requiring $r_N(x)$ to be zero at some $\{z_j\}_{1\le j\le N}$, where $z_j\in[0,A]$ for all $j=1,2,\ldots,N$, one can achieve a certain level of proximity of the residual to zero. As a result, one readily arrives at the following system of $N$ algebraic equations on the coefficients $u_{j,N}$
\begin{align}\label{eq:fredholm-lin-system}
\boldsymbol{u}_N
&=
\boldsymbol{\upsilon}
+
\mathcal{\boldsymbol{K}}_\infty \boldsymbol{u}_N,
\end{align}
where $\boldsymbol{u}_N=[u_{1,N},\ldots,u_{N,N}]^\top$, $\boldsymbol{\upsilon}=[\upsilon(z_1),\ldots,\upsilon(z_N)]^\top$, and  $\boldsymbol{K}_\infty$ is a matrix of size $N$-by-$N$ such that
\begin{align}\label{eq:def-matrix-K}
(\boldsymbol{K}_\infty)_{i,j}
&=
\int_0^A K_\infty(z_i,y)\,\phi_j(y)\,dy,\;\;1\le i,j\le N.
\end{align}
The points $\{z_j\}_{1\le j\le N}$ are referred to as the collocation points.

We now remark that for the system of linear equations~\eqref{eq:fredholm-lin-system} to be consistent, the functions $\{\phi_j(x)\}_{1\le j\le N}$ need to be chosen so as to form a basis in the appropriate functional space, i.e., in particular, $\{\phi_j(x)\}_{1\le j\le N}$ need to be linearly independent. This is equivalent to introducing an interpolating projection operator, $\pi_N$, that projects the sought function $u(x)$ onto the span of these basis functions. Specifically, this operator is defined as $\pi_N\circ u=\sum_{j=1}^N u_{j,N} \phi_j(x)$, and
\begin{align*}
\|\pi_N\|_\infty
&=
\max_{0\le x\le A}\sum_{j=1}^N|\phi_j(x)|\ge1.
\end{align*}

It is also worth noting that, by design, this method is most accurate at the collocation points, $\{z_j\}_{1\le j\le N}$, i.e., at the points at which the residual is zero. For an arbitrary point $x$ the unknown function can be evaluated as
\begin{align}\label{eq:int-eqn-iter-sol}
\begin{aligned}
\widetilde{u}_N(x)
&=
\upsilon(x)+\int_0^A K_\infty(x,y)\,u_N(y)\,dy\\
&=
\upsilon(x)+\sum_{j=1}^Nu_{j,N}\int_0^A K_\infty(x,y)\,\phi_j(y)\,dy.
\end{aligned}
\end{align}
This technique is known as the iterated solution; see, e.g.,~\cite[Subsection~12.3.2]{Atkinson+Han:Book09}. Also note that $\widetilde{u}_N(z_j)=u_N(z_j)=u_{j,N}$, $1\le j\le N$.

We now consider the question of accuracy. To that end, it is apparent that the choice of the basis functions, $\{\phi_i(x)\}_{1\le i\le N}$, must play a critical role. This is, in fact, the case, as may be concluded from, e.g.,~\cite[Theorem~12.1.12,~p.~479]{Atkinson+Han:Book09}. Specifically, using $\|u-\widetilde{u}_N\|_\infty$ as a sensible measure of the method's error and applying~\cite[Formula~12.3.21,~p.~499]{Atkinson+Han:Book09}, we obtain
\begin{align}\label{eq:err-bound-1}
\begin{aligned}
\|u-\widetilde{u}_N\|_\infty
&\le
\|(\mathbb{I}-\mathcal{K}_\infty)^{-1}\|_\infty\|\mathcal{K}_\infty\circ(\mathbb{I}-\pi_N)\circ u\|_\infty\\
&\le
\|(\mathbb{I}-\mathcal{K}_\infty)^{-1}\|_\infty\|\mathcal{K}_\infty\|_\infty\|(\mathbb{I}-\pi_N)\circ u\|_\infty,
\end{aligned}
\end{align}
whence one can see that the method's error is determined by two factors:\begin{inparaenum}[\itshape a)]\item $\|(\mathbb{I}-\mathcal{K}_\infty)^{-1}\|_\infty$ and \item $\|(\mathbb{I}-\pi_N)\circ u\|_\infty$, which is the corresponding interpolation error\end{inparaenum}. The latter can be found for each particular choice of $\pi_N$. The bigger problem is to upperbound $\|(\mathbb{I}-\mathcal{K}_\infty)^{-1}\|_\infty$. To this end, the standard result
\begin{align*}
\|(\mathbb{I}-\mathcal{K}_\infty)^{-1}\|_\infty
&\le
\dfrac{1}{1-\|\mathcal{K}_\infty\|_\infty}
\end{align*}
is applicable since $\|\mathcal{K}_\infty\|_\infty<1$. However, it is well-known that this is a very generic result, and as an upper bound is often far off. Consequently, this loosens the upper bound on $\|u-\widetilde{u}_N\|_\infty$ and diminishes its practical potential. However, since in our particular case ${K}_\infty(x,y)$ is a transition probability kernel, a tighter (in fact, exact) bound is possible to obtain. We now state the corresponding result.
\begin{lemma}\label{lem:ARL-inv-bound}
$\|\,(\mathbb{I}-\mathcal{K}_\infty)^{-1}\|_\infty=\|\,\ell\,\|_\infty$.
\end{lemma}
\begin{proof}
The desired result can be shown using the argument akin to one given to prove, e.g.,~\cite[Theorem~5, part (b)]{Fan:MfM1958} for square matrices with non-negative elements. Specifically, even though $\mathcal{K}_\infty$ and $(\mathbb{I}-\mathcal{K}_\infty)^{-1}$ are not matrices, Fan's proof for matrices extends easily to $\mathcal{K}_\infty$ and $(\mathbb{I}-\mathcal{K}_\infty)^{-1}$, as either operator is bounded, linear and non-negative-valued, since $K_\infty(x,y)\ge0$ for all $x,y\ge0$.
\end{proof}

This lemma allows to find \emph{exactly} the bound for the proposed method's error. We now state the corresponding result for equation~\eqref{eq:ARL-int-eqn}.
\begin{theorem}\label{thm:error-bound}
Assuming $N\ge1$ is sufficiently large
\begin{align*}
\|\ell-\widetilde{\ell}_N\|_\infty
&<
\|\ell\|_\infty\|\ell_{xx}\|_\infty\frac{h^2}{8},\;\text{where}\; \ell_{xx}\triangleq\dfrac{\partial^2}{\partial x^2}\ell(x,A);
\end{align*}
note that the inequality is strict.
\end{theorem}
\begin{proof}
It suffices to recall the standard bound on the polynomial interpolation error
\begin{align*}
\|(\mathbb{I}-\pi_N)\circ \ell\|_\infty
&\leq
\|\ell_{xx}\|_\infty\frac{h^2}{8};
\end{align*}
see, e.g.,~\cite[Formula~3.2.9, p.~124]{Atkinson+Han:Book09}. The desired inequality can then be readily obtained from~\eqref{eq:err-bound-1}, Lemma~\ref{lem:ARL-inv-bound}, and the fact that $\|\mathcal{K}\|_\infty<1$.
\end{proof}

We propose to use piecewise linear polynomial space. To this end, for any positive integer $N\ge2$, let $\Pi_N\colon 0\triangleq x_0<x_1<\ldots<x_{N-1}=A$ denote a partition of the interval $[0,A]$, and for $i=1,\ldots,N-1$ set $I_i^N\triangleq(x_{i-1}, x_i)$, $h_i\triangleq x_i-x_{i-1}(>0)$, and $h\triangleq h(N)=\max_{1\le i\le N-1} h_i$; assume also that $h\to0$ as $N\to\infty$.
Next, set $z_j=x_{j-1}$, $1\le j\le N$ and use the ``hat functions''
\begin{empheq}[%
    left={%
        \phi_i(x)=%
    \empheqlbrace}]{align}\nonumber
&\frac{x-x_{i - 2}}{h_{i-1}},\quad\text{if $x\in I_{i-1}^N, i>1$;}\label{eq:lin-basis}\\
&\frac{x_i-x}{h_{i}},\quad\text{if $x\in I_{i}^N, i<N$;}\\\nonumber
&0,\quad\text{otherwise},
\end{empheq}
where $1\le i\le N$. These functions are clearly linearly independent. One important remark is now in order to justify this choice. The error bound given in Theorem~\ref{thm:error-bound} is tight, and can be seen to be directly proportional to the magnitude of the solution, i.e., to $\ell(x,A)=\ARL(\mathcal{S}_A^x)$. Worse yet, it is also proportional to the detection threshold squared lurking in $h^2$. Since $\ARL(\mathcal{S}_A^x)\approx A/\xi-x$, and $\xi\in(0,1)$, one can roughly set $A\approx\ell(x,A)$ and conclude that the error bound is roughly proportional to the magnitude of the solution cubed. This may seem to drastically offset the second power of $N$ in the denominator buried in $h^2$. However, as we will demonstrate experimentally, due to the almost linearity of $\ell(x,A)$ in $x$ at least for $x\in[0,A]$, this does not happen. Specifically, since $\ell(x,A)\approx A/\xi-x$ at least for $x\in[0,A]$, we have:
\begin{align*}
\frac{\partial^2}{\partial x^2}
\ell(x,A)
&\approx 0,\;\text{at least for}\; x\in[0,A].
\end{align*}
This makes the obtained error bound extremely close to zero, even for relatively small $N$. Consequently, the method is rather accurate and does not require $N$ to be large. This is precisely the reason why we suggested to use the piecewise linear basis.

The only question left to discuss is that of the role of the change-of-measure ploy in all this. To this end, recall that to implement the numerical method, one has to compute the matrix~\eqref{eq:def-matrix-K}, which involves finding the corresponding integrals. At first, it may seem to be an additional source of errors, since the integrals would have to be evaluated numerically as well. However, due to the change-of-measure identity, $(1+x){K}_0(x,y)=y{K}_\infty(x,y)$, the integrals involved in~\eqref{eq:def-matrix-K} can be computed {\it exactly}. Specifically, using~\eqref{eq:def-matrix-K} and~\eqref{eq:lin-basis}, and recalling that $z_j=x_{j-1}$, $1\le j\le N$, the corresponding formula is as follows:
\begin{align*}
(\boldsymbol{K}_\infty)_{i,j}
&=
\int_0^A{K}_\infty(x_i,y)\,\phi_j(y)\,dy\\
&=
\dfrac{1}{h_{j-1}}\left\{(1+x_i)\left[{P}_0^{\LR}\left(\dfrac{x_{j-1}}{1+x_i}\right)-{P}_0^{\LR}\left(\dfrac{x_{j-2}}{1+x_i}\right)\right]\right.-\\
&\qquad\qquad\qquad\qquad\qquad x_{j-2}\left.\left[{P}_\infty^{\LR}\left(\dfrac{x_{j-1}}{1+x_i}\right)-{P}_\infty^{\LR}\left(\dfrac{x_{j-2}}{1+x_i}\right)\right]\right\}\indicator{j>1}+\\
&\qquad\dfrac{1}{h_{j}}\left\{x_j\left[{P}_\infty^{\LR}\left(\dfrac{x_{j}}{1+x_i}\right)-{P}_\infty^{\LR}\left(\dfrac{x_{j-1}}{1+x_i}\right)\right]\right.-\\
&\qquad\qquad\qquad\qquad\qquad (1+x_i)\left.\left[{P}_0^{\LR}\left(\dfrac{x_{j}}{1+x_i}\right)-{P}_0^{\LR}\left(\dfrac{x_{j-1}}{1+x_i}\right)\right]\right\}\indicator{j<N}
\end{align*}
for $1\le i,j\le N$.

To wrap this subsection, note that the developed numerical framework can also be used to assess the accuracy of the popular Markov chain approach, introduced by~\cite{Brook+Evans:B1972}, and later extended, e.g., by~\cite{Woodall:T1983}. To this end, as noted by~\cite{Champ+Rigdon:CommStat1991}, the Markov chain approach is equivalent to the integral-equations approach if the integral is approximated via the product midpoint rule. This, in turn, is equivalent to choosing the basis functions, $\{\phi_i(x)\}_{1\le j\le N}$, as piecewise constants on $\Pi_{N+1}$, i.e., $\phi_j(x)=\indicator{x\in I_j^{N+1}}$, and equating the residual to zero at the midpoints of the intervals $I_j^{N+1}$, i.e., setting $z_j=(x_{j-1}+x_j)/2$, $1\le j\le N$. In this case the $(i,j)$-th element of the matrix $\boldsymbol{K}$ defined by~\eqref{eq:def-matrix-K} is as follows:
\begin{align*}
(\boldsymbol{K}_\infty)_{i,j}
&=
P_\infty^{\LR}\left(\dfrac{x_j}{1+z_i}\right)
-
P_\infty^{\LR}\left(\dfrac{x_{j-1}}{1+z_i}\right),\;1\le i,j\le N.
\end{align*}
It can be shown that this approach exhibits a superconvergence effect: the rate is also quadratic, even though the interpolation is based on using polynomials of degree zero (i.e., constants). This was first observed by~\cite{Kryloff+Bogoliubov:CRA-URSS1929}. See also~\cite[pp.~130--135]{Kantorovich+Krylov:Book1958} for a different proof. However, in spite of the quadratic convergence, the corresponding constant in the error bound is large, and, as a result, the partition size required by the method ends up being substantial. In fact, this method was employed, e.g., by~\cite{Moustakides+etal:CommStat09}, to compare the CUSUM chart and the original SR procedure, and the partition size they used consisted of thousands of points to ensure reasonable accuracy. In the next section we will compare this method against ours, and show that our method is superior.

Next, we apply the proposed numerical method to a particular example to study the GSR procedure's pre-change Run-Length distribution. We note that, to the best of our knowledge, this is the first time the GSR procedure's actual pre-change Run-Length distribution is looked at.

%-------------------------------------------------------------------------------------------------%
\section{A case study}
\label{sec:case-study}

Consider a scenario where the observations are independent Gaussian with mean zero pre-change and $\theta\neq0$ post-change; the variance is 1 (known) and does not change. Formally, the pre- and post-change densities in this case are
\begin{align*}
f(x)
&=
\dfrac{1}{\sqrt{2\pi}}\exp\left\{-\dfrac{x^2}{2}
\right\}\;\text{and}\;
g(x)=\dfrac{1}{\sqrt{2\pi}}\exp\left\{-\dfrac{(x-\theta)^2}{2}\right\},
\end{align*}
respectively, where $x\in\mathbb{R}$ and $\theta\neq0$.

Next, it is straightforward to see that
\begin{align*}
\LR_n
&\triangleq
\dfrac{g(X_n)}{f(X_n)}
=\exp\left\{\theta X_n-\frac{\theta^2}{2}\right\},\;n\ge1,
\end{align*}
and subsequently, $P_d^{\LR}(t)\triangleq\Pr_d(\LR_1\le t)$ is log-normal with mean $-\theta^2/2$ and variance $\theta^2$ under $d=\infty$, and with mean $\theta^2/2$ and variance $\theta^2$ under $d=0$. Hence, one obtains
\begin{align*}
K_\infty(x,y)
&\triangleq
\dfrac{\partial}{\partial y}P_\infty^{\LR}\left(\dfrac{y}{1+x}\right)
=
\dfrac{1}{y\sqrt{2\pi\theta^2}}\exp\left\{-\dfrac{1}{2\theta^2}\left(\log\dfrac{y}{1+x}+\dfrac{\theta^2}{2}\right)^2\right\}\indicator{y/(1+x)\ge0},
\end{align*}
and one can see that the kernel is the same regardless of whether $\theta<0$ or $\theta>0$. We therefore, without any loss of generality, will consider only the former case, i.e., assume from now on that $\theta>0$. As $(1+x)K_0(x,y)=y K_\infty(x,y)$, one could quickly obtain the formula for $K_0(x,y)$, if it were necessary. However, it isn't, since the change-of-measure identity allows to express everything in terms of $K_\infty(x,y)$ only.

We now employ the proposed numerical method to evaluate the characteristics of the GSR procedure. The method was implemented in MATLAB. To test the accuracy of the method, we will perform our analysis for a broad range of changes $\theta=0.01, 0.1, 0.5$ and $1.0$, i.e., from very faint, to small, to moderate, to very contrast, respectively. We will also consider a wide range of values of the ARL to false alarm: $\gamma=10^2, 10^3, 10^4$ and $10^5$, i.e., from large risk of false alarm, to moderate, to low, respectively (the false alarm risk is inversely proportional to the ARL to false alarm).

Since the accuracy of our numerical method is determined by the accuracy of the underlying piecewise linear polynomial interpolation, we recall that forming the partition of the interval of integration, $[0,A]$, by the Chebyshev nodes (i.e., roots of the Chebyshev polynomials of the first kind) will result in the smallest possible interpolating error. Subsequently, the overall accuracy of the method will improve as well. Specifically, we partition the interval $[0,A]$ into $N-1$, $N\ge2$, non-overlapping subintervals joint at the shifted Chebyshev nodes of the form:
\begin{align*}
x_{N-i}
&=
\frac{A}{2}\left\{1+\cos\left[(2i - 1)\dfrac{\pi}{2N}\right]/\cos\left(\dfrac{\pi}{2N}\right)\right\},\;1\le i\le N.
\end{align*}
The shift is to make sure the left and right end nodes coincide with the left and right end points of the interval $[0,A]$, i.e., $x_0=0$ and $x_{N-1}=A$. To evaluate the solution at a point different from one of $x_j$'s, we will use the iterated solution~\eqref{eq:int-eqn-iter-sol}. As we will see below, the method's accuracy is high even if $N$ is small, and the accuracy remains high for a large range of $\theta$'s and for a large range of values of the ARL to false alarm. To measure the rate of convergence of the method we will rely on the common Richardson extrapolation technique: if $\ell_{2N}$, $\ell_N$ and $\ell_{N/2}$ are the solutions assuming the partition size is $2N$, $N$ and $N/2$, respectively, then the corresponding rate of convergence, $p$, can be estimated as
\begin{align*}
2^{-p(N)}
&\approx
\dfrac{\|\ell_{2N}-\ell_N\|_\infty}{\|\ell_{N}-\ell_{N/2}\|_\infty},\;\text{so that}\;
p(N)\approx-\log_2\dfrac{\|\ell_{2N}-\ell_N\|_\infty}{\|\ell_{N}-\ell_{N/2}\|_\infty}.
\end{align*}
We will try $N=2^i$ for $i=1,2,\ldots,12$. Also, the actual error can be estimated by using $\|\ell-\ell_N\|_\infty\approx2^{-p(N)}\|\ell_{N}-\ell_{N/2}\|_\infty$.

The first set of numerical results from the convergence analysis is summarized in Tables~\ref{tab:ARL__theta__0_01},~\ref{tab:ARL__theta__0_1},~\ref{tab:ARL__theta__0_5}, and~\ref{tab:ARL__theta__1_0}. Specifically, each table is for a specific value of $\theta$, and reports the estimated convergence rate for various values of the ARL to false alarm for two methods: one proposed in this work and one proposed in the work of~\cite{Moustakides+etal:SS11}. For the latter, the corresponding results are shown in brackets with [NaN]'s indicating the method's failure. We observe that, as expected from Theorem~\ref{thm:error-bound}, the rate of convergence of our method is quadratic. More importantly, it is attained almost right away, regardless of the magnitude of the change or the value of the ARL to false alarm. The reason is the use of the linear interpolation to approximate the solution of the corresponding integral equation, which is almost linear with respect to the headstart. Hence, one would expect our numerical scheme to work well even for small $N$. The results presented in the tables do confirm that. On the other hand, the method of~\cite{Moustakides+etal:SS11} can be seen to converge much slower, requiring $N$ to be in the thousands or higher to deliver decent accuracy.
\begin{sidewaystable}
    \centering
    \begin{tabularx}{\textwidth}{c *{9}{Y}}
    \toprule

     & \multicolumn{2}{c}{$\gamma=10^2$ ($A=99.2$)}
     & \multicolumn{2}{c}{$\gamma=10^3$ ($A=994.2$)}
     & \multicolumn{2}{c}{$\gamma=10^4$ ($A=9941.9$)}
     & \multicolumn{2}{c}{$\gamma=10^5$ ($A=99419.0$)}\\
    \cmidrule(lr){2-3} \cmidrule(lr){4-5} \cmidrule(lr){6-7} \cmidrule(lr){8-9}
          $N$ & $\ARL(\mathcal{S}_A)$ & Rate & $\ARL(\mathcal{S}_A)$ & Rate & $\ARL(\mathcal{S}_A)$ & Rate & $\ARL(\mathcal{S}_A)$ & Rate\\
    \midrule
        \multirow{2}{*}{2}&100.48747&-&1,002.54172&-&10,021.90967&-&100,215.7892&-\\
        &[NaN]&-&[NaN]&-&[NaN]&-&[NaN]&-\\
        \midrule
        \multirow{2}{*}{4}&100.48747&0.37551&1,002.54172&0.68221&10,021.90967&0.77378&100,215.7892&0.34451\\
        &[NaN]&[NaN]&[NaN]&[NaN]&[NaN]&[NaN]&[NaN]&[NaN]\\
        \midrule
        \multirow{2}{*}{8}&100.48747&-34.82763&1,002.54172&-32.94567&10,021.90967&-38.09409&100,215.7892&-27.23368\\
        &[NaN]&[NaN]&[NaN]&[NaN]&[NaN]&[NaN]&[NaN]&[NaN]\\
        \midrule
        \multirow{2}{*}{16}&100.22051&-0.72003&1,001.72125&-1.39624&10,014.54917&-1.44092&100,142.96209&-1.44092\\
        &[NaN]&[NaN]&[NaN]&[NaN]&[NaN]&[NaN]&[NaN]&[NaN]\\
        \midrule
        \multirow{2}{*}{32}&100.10203&2.81168&1,000.52457&2.57269&10,002.81619&2.55093&100,025.87105&2.54878\\
        &[NaN]&[NaN]&[NaN]&[NaN]&[NaN]&[NaN]&[NaN]&[NaN]\\
        \midrule
        \multirow{2}{*}{64}&100.08052&2.09695&1,000.33233&2.00637&10,000.90452&1.99517&100,006.76431&1.99403\\
        &[89.6219]&-&[NaN]&[NaN]&[NaN]&[NaN]&[NaN]&[NaN]\\
        \midrule
        \multirow{2}{*}{128}&100.07523&2.02038&1,000.28277&1.99916&10,000.4097&1.99661&100,001.81661&1.99635\\
        &[104.01353]&[10.96141]&[NaN]&[NaN]&[NaN]&[NaN]&[NaN]&[NaN]\\
        \midrule
        \multirow{2}{*}{256}&100.07391&2.00506&1,000.27031&1.99966&10,000.28518&1.99901&100,000.57131&1.99894\\
        &[99.57944]&[-6.54683]&[NaN]&[NaN]&[NaN]&[NaN]&[NaN]&[NaN]\\
        \midrule
        \multirow{2}{*}{512}&100.07358&2.00126&1,000.2672&1.99991&10,000.25399&1.99974&100,000.25943&1.99973\\
        &[100.11896]&[8.90420]&[995.50653]&-&[NaN]&[NaN]&[NaN]&[NaN]\\
        \midrule
        \multirow{2}{*}{1024}&100.0735&2.00032&1,000.26642&1.99998&10,000.24619&1.99993&100,000.18142&1.99993\\
        &[100.07442]&[-2.61104]&[1,000.52276]&[9.90705]&[NaN]&[NaN]&[NaN]&[NaN]\\
        \midrule
        \multirow{2}{*}{2048}&100.07348&2.00008&1,000.26622&1.99999&10,000.24424&1.99998&100,000.16191&1.99999\\
        &[100.07813]&[0.38884]&[1,000.3918]&[1.47489]&[NaN]&[NaN]&[NaN]&[NaN]\\
        \midrule
        \multirow{2}{*}{4096}&100.07347&-&1,000.26617&-&10,000.24375&-&100,000.15704&-\\
        &[100.07324]&-&[1,000.31815]&-&[NaN]&-&[NaN]&-\\
    \bottomrule
    \end{tabularx}
    \caption{Convergence analysis for $\ARL(\mathcal{S}_A)$ assuming $\theta=0.01$.}
    \label{tab:ARL__theta__0_01}
\end{sidewaystable}
\begin{sidewaystable}
    \centering
    \begin{tabularx}{\textwidth}{c *{9}{Y}}
    \toprule

     & \multicolumn{2}{c}{$\gamma=10^2$ ($A=94.34$)}
     & \multicolumn{2}{c}{$\gamma=10^3$ ($A=943.41$)}
     & \multicolumn{2}{c}{$\gamma=10^4$ ($A=9434.08$)}
     & \multicolumn{2}{c}{$\gamma=10^5$ ($A=94340.5$)}\\
    \cmidrule(lr){2-3} \cmidrule(lr){4-5} \cmidrule(lr){6-7} \cmidrule(lr){8-9}
          $N$ & $\ARL(\mathcal{S}_A)$ & Rate & $\ARL(\mathcal{S}_A)$ & Rate & $\ARL(\mathcal{S}_A)$ & Rate & $\ARL(\mathcal{S}_A)$ & Rate\\
    \midrule
        \multirow{2}{*}{2}&102.58157&-&1,022.1717&-&10,218.14169&-&102,177.54835&-\\
        &[$\approx 4\times 10^{10}$]&-&[$\approx 5.2\times 10^{11}$]&-&[$\approx 6.8\times 10^{11}$]&-&[$\approx 6.9\times 10^{11}$]&-\\
        \midrule
        \multirow{2}{*}{4}&102.55852&-7.36401&1,022.00734&-7.69785&10,216.55311&-7.73156&102,161.71661&-7.73494\\
        &[$\approx 2.9\times 10^{9}$]&[8.39767]&[$\approx 3.9\times 10^{11}$]&[2.50739]&[$\approx 6.7\times 10^{11}$]&[-0.22883]&[$\approx 7\times 10^{11}$]&[-0.72388]\\
        \midrule
        \multirow{2}{*}{8}&100.72368&2.24005&1,004.66098&2.19119&10,044.04284&2.18632&100,437.56746&2.18583\\
        &[$\approx 3.3\times 10^{7}$]&[13.16679]&[$\approx 2.2\times 10^{11}$]&[0.81472]&[$\approx 6.4\times 10^{11}$]&[-1.53456]&[$\approx 7.2\times 10^{11}$]&[-0.97868]\\
        \midrule
        \multirow{2}{*}{16}&100.39247&1.96637&1,001.3746&1.95635&10,011.20082&1.95538&100,109.16953&1.95529\\
        &[$\approx 4.1\times 10^4$]&[7.54089]&[$\approx 6.9\times 10^{10}$]&[2.01676]&[$\approx 5.9\times 10^{11}$]&[0.12059]&[$\approx 7.5\times 10^{11}$]&[-1.02206]\\
        \midrule
        \multirow{2}{*}{32}&100.3113&2.00152&1,000.55799&1.99687&10,003.02932&1.99639&100,027.44935&1.99635\\
        &[128.34953]&[2.03617]&[$\approx 8\times 10^{9}$]&[5.07392]&[$\approx 4.9\times 10^{11}$]&[3.77045]&[$\approx 8\times 10^{11}$]&[-0.93099]\\
        \midrule
        \multirow{2}{*}{64}&100.29088&2.00006&1,000.35206&1.99884&10,000.96814&1.99872&100,006.83578&1.9987\\
        &[99.87488]&[-3.6633]&[$\approx 1.8\times 10^{8}$]&[10.40265]&[$\approx 3.2\times 10^{11}$]&[-1.53804]&[$\approx 8.8\times 10^{11}$]&[-0.63049]\\
        \midrule
        \multirow{2}{*}{128}&100.28576&1.99999&1,000.30045&1.99968&10,000.45155&1.99965&100,001.66941&1.99965\\
        &[100.24439]&[3.23903]&[$\approx 4.3\times 10^5$]&[14.55815]&[$\approx 1.3\times 10^{11}$]&[1.01609]&[$\approx 10^{12}$]&[-0.01633]\\
        \midrule
        \multirow{2}{*}{256}&100.28448&2.0&1,000.28754&1.99992&10,000.32232&1.99991&100,000.37698&1.99991\\
        &[100.3111]&[2.2709]&[1,249.19154]&[3.06013]&[$\approx 2\times 10^{10}$]&[3.82822]&[$\approx 1.1\times 10^{12}$]&[7.72911]\\
        \midrule
        \multirow{2}{*}{512}&100.28416&2.0&1,000.28431&1.99998&10,000.29001&1.99998&100,000.05382&1.99998\\
        &[100.28414]&[9.51626]&[1,000.05942]&[-5.63427]&[$\approx 7.6\times 10^{8}$]&[8.34186]&[$\approx 9.2\times 10^{11}$]&[-4.70896]\\
        \midrule
        \multirow{2}{*}{1024}&100.28408&2.0&1,000.2835&1.99999&10,000.28193&1.99999&99,999.97303&2.00001\\
        &[100.28415]&[-1.36052]&[1,000.14412]&[1.30463]&[$\approx 3.8\times 10^{6}$]&[15.53288]&[$\approx 4.5\times 10^{11}$]&[0.26133]\\
        \midrule
        \multirow{2}{*}{2048}&100.28406&2.0&1,000.2833&2.00000&10,000.27991&2.00000&99,999.95283&1.99995\\
        &[100.28408]&[2.0]&[1,000.23442]&[1.93052]&[12,725.65581]&[9.49778]&[$\approx 7.4\times 10^{10}$]&[2.79838]\\
        \midrule
        \multirow{2}{*}{4096}&100.28406&-&1,000.28325&-&10,000.27941&-&99,999.94779&-\\
        &[100.28406]&-&[1,000.28103]&-&[10,003.04331]&-&[$\approx 3\times 10^{9}$]]&-\\
    \bottomrule
    \end{tabularx}
    \caption{Convergence analysis for $\ARL(\mathcal{S}_A)$ assuming $\theta=0.1$.}
    \label{tab:ARL__theta__0_1}
\end{sidewaystable}
\begin{sidewaystable}
    \centering
    \begin{tabularx}{\textwidth}{c *{9}{Y}}
    \toprule

     & \multicolumn{2}{c}{$\gamma=10^2$ ($A=74.76$)}
     & \multicolumn{2}{c}{$\gamma=10^3$ ($A=747.62$)}
     & \multicolumn{2}{c}{$\gamma=10^4$ ($A=7476.15$)}
     & \multicolumn{2}{c}{$\gamma=10^5$ ($A=74761.5$)}\\
    \cmidrule(lr){2-3} \cmidrule(lr){4-5} \cmidrule(lr){6-7} \cmidrule(lr){8-9}
          $N$ & $\ARL(\mathcal{S}_A)$ & Rate & $\ARL(\mathcal{S}_A)$ & Rate & $\ARL(\mathcal{S}_A)$ & Rate & $\ARL(\mathcal{S}_A)$ & Rate\\
    \midrule
        \multirow{2}{*}{2}&112.08343&-&1,115.94875&-&11,154.51566&-&111,540.26117&-\\
        &[40.48187]&-&[50.02797]&-&[51.15832]&-&[51.27337]&-\\
        \midrule
        \multirow{2}{*}{4}&102.84839&2.23395&1,024.79306&2.20209&10,244.14193&2.19892&102,437.69936&2.1986\\
        &[70.3026]&[0.17979]&[109.88631]&[-0.9552]&[115.49826]&[-1.07609]&[116.08141]&[-1.08829]\\
        \midrule
        \multirow{2}{*}{8}&100.9572&2.18209&1,005.64335&2.18298&10,052.41311&2.18308&100,520.1782&2.18309\\
        &[96.62922]&[1.81075]&[225.94273]&[-0.75784]&[251.147]&[-1.00629]&[253.8778]&[-1.03152]\\
        \midrule
        \multirow{2}{*}{16}&100.57276&2.0067&1,001.74986&2.0052&10,013.43&2.00504&100,130.29863&2.00502\\
        &[104.13341]&[2.86321]&[422.18974]&[-0.45992]&[523.63068]&[-0.95905]&[535.55853]&[-1.01016]\\
        \midrule
        \multirow{2}{*}{32}&100.47686&2.00153&1,000.77749&2.0011&10,003.69296&2.00106&100,032.91489&2.00106\\
        &[101.45615]&[2.02681]&[692.12063]&[0.10653]&[1,053.34664]&[-0.89331]&[1,102.90019]&[-0.99597]\\
        \midrule
        \multirow{2}{*}{64}&100.45288&2.00037&1,000.53434&2.00026&10,001.25812&2.00025&100,008.56303&2.00025\\
        &[100.52541]&[6.07368]&[942.83838]&[1.37428]&[2,037.25903]&[-0.77560]&[2,234.4169]&[-0.98059]\\
        \midrule
        \multirow{2}{*}{128}&100.44689&2.00009&1,000.47355&2.00007&10,000.64937&2.00006&100,002.47469&2.00006\\
        &[100.44801]&[1.84739]&[1,039.55157]&[3.70162]&[3,721.61729]&[-0.54838]&[4,467.20008]&[-0.95572]\\
        \midrule
        \multirow{2}{*}{256}&100.44539&2.00002&1,000.45836&2.00002&10,000.49718&2.00002&100,000.95257&2.00002\\
        &[100.44504]&[3.9197]&[1,016.14038]&[0.72821]&[6,184.8977]&[-0.0999]&[8,797.78196]&[-0.90889]\\
        \midrule
        \multirow{2}{*}{512}&100.44501&2.00001&1,000.45456&2.0&10,000.45913&2.0&100,000.57205&2.0\\
        &[100.44497]&[0.61014]&[1,001.12312]&[4.70144]&[8,824.79151]&[0.84811]&[16,928.882]&[-0.81709]\\
        \midrule
        \multirow{2}{*}{1024}&100.44492&2.0&1,000.45361&2.0&10,000.44962&2.0&100,000.4769&2.00002\\
        &[100.44491]&[1.9999]&[1,000.4463]&[2.95142]&[10,291.28225]&[7.04837]&[31,254.63194]&[-0.63583]\\
        \midrule
        \multirow{2}{*}{2048}&100.4449&2.0&1,000.45337&2.0&10,000.44724&2.0&100,000.4531&1.99992\\
        &[100.44489]&[2.0]&[1,000.4529]&[5.77463]&[10,257.2434]&[-2.68232]&[53,514.38599]&[-0.27727]\\
        \midrule
        \multirow{2}{*}{4096}&100.44489&-&1,000.45331&-&10,000.44665&-&100,000.44718&-\\
        &[100.44489]&-&[1,000.4533]&-&[10,023.86256]&-&[80,490.89164]&-\\
    \bottomrule
    \end{tabularx}
    \caption{Convergence analysis for $\ARL(\mathcal{S}_A)$ assuming $\theta=0.5$.}
    \label{tab:ARL__theta__0_5}
\end{sidewaystable}
\begin{sidewaystable}
    \centering
    \begin{tabularx}{\textwidth}{c *{9}{Y}}
    \toprule
     & \multicolumn{2}{c}{$\gamma=10^2$ ($A=56.0$)}
     & \multicolumn{2}{c}{$\gamma=10^3$ ($A=560.0$)}
     & \multicolumn{2}{c}{$\gamma=10^4$ ($A=5603.5$)}
     & \multicolumn{2}{c}{$\gamma=10^5$ ($A=56037.0$)}\\
    \cmidrule(lr){2-3} \cmidrule(lr){4-5} \cmidrule(lr){6-7} \cmidrule(lr){8-9}
          $N$ & $\ARL(\mathcal{S}_A)$ & Rate & $\ARL(\mathcal{S}_A)$ & Rate & $\ARL(\mathcal{S}_A)$ & Rate & $\ARL(\mathcal{S}_A)$ & Rate\\
    \midrule
        \multirow{2}{*}{2}&126.30518&-&1,255.83903&-&12,558.81029&-&125,585.16094&-\\
        &[18.90254]&-&[21.19442]&-&[21.4518]&-&[21.47784]&-\\
        \midrule
        \multirow{2}{*}{4}&102.91218&3.73911&1,021.53276&3.80202&10,214.35752&3.80672&102,139.87197&3.80718\\
        &[36.13905]&[-0.44375]&[46.03212]&[-0.96967]&[47.27734]&[-1.03168]&[47.40492]&[-1.038]\\
        \midrule
        \multirow{2}{*}{8}&101.36866&1.73324&1,007.01242&1.58154&10,069.74031&1.56622&100,694.31777&1.56474\\
        &[59.58298]&[0.04205]&[94.67408]&[-0.88292]&[100.07532]&[-1.00644]&[100.64298]&[-1.01923]\\
        \midrule
        \multirow{2}{*}{16}&100.8784&2.04648&1,001.71281&2.13084&10,016.23849&2.14898&100,158.83358&2.15053\\
        &[82.35344]&[0.74187]&[184.37489]&[-0.74321]&[206.14337]&[-0.98293]&[208.5477]&[-1.00859]\\
        \midrule
        \multirow{2}{*}{32}&100.76017&2.00161&1,000.52849&1.98007&10,004.47337&1.96408&100,041.24097&1.96265\\
        &[95.96936]&[1.6784]&[334.52493]&[-0.49879]&[415.78372]&[-0.94804]&[425.64534]&[-0.99915]\\
        \midrule
        \multirow{2}{*}{64}&100.73062&2.00075&1,000.22665&2.00270&10,001.43135&2.01008&100,010.79695&2.0115\\
        &[100.22338]&[2.99676]&[546.69114]&[-0.09267]&[820.23101]&[-0.8899]&[859.58367]&[-0.99075]\\
        \midrule
        \multirow{2}{*}{128}&100.72324&2.00019&1,000.15136&2.00010&10,000.67756&1.99883&100,003.2628&1.99767\\
        &[100.75632]&[5.36654]&[772.93304]&[0.5279]&[1,569.69277]&[-0.78224]&[1,721.91099]&[-0.97816]\\
        \midrule
        \multirow{2}{*}{256}&100.72139&2.00005&1,000.13254&2.00004&10,000.4889&2.00011&100,001.37516&2.00056\\
        &[100.73537]&[1.02544]&[929.84595]&[1.38532]&[2,858.62024]&[-0.58435]&[3,420.64864]&[-0.95481]\\
        \midrule
        \multirow{2}{*}{512}&100.72093&2.00001&1,000.12783&2.00001&10,000.44174&2.00001&100,000.90348&1.99996\\
        &[100.72228]&[3.26059]&[989.91309]&[2.54712]&[4,791.19095]&[-0.24044]&[6,713.34755]&[-0.90975]\\
        \midrule
        \multirow{2}{*}{1024}&100.72082&2.0&1,000.12666&2.0&10,000.42995&2.0&100,000.78555&2.00001\\
        &[100.72085]&[4.39192]&[1,000.19036]&[6.04614]&[7,074.24421]&[0.30588]&[12,899.39463]&[-0.82298]\\
        \midrule
        \multirow{2}{*}{2048}&100.72079&2.0&1,000.12636&2.0&10,000.427&1.99999&100,000.75607&1.99996\\
        &[100.72078]&[4.19398]&[1,000.34589]&[0.63957]&[8,921.11562]&[1.08545]&[23,842.84271]&[-0.66023]\\
        \midrule
        \multirow{2}{*}{4096}&100.72078&-&1,000.12629&-&10,000.42626&-&100,000.7487&-\\
        &[100.72078]&-&[1,000.15503]&-&[9,791.44373]&-&[41,137.15155]&-\\
    \bottomrule
    \end{tabularx}
    \caption{Convergence analysis for $\ARL(\mathcal{S}_A)$ assuming $\theta=1.0$.}
    \label{tab:ARL__theta__1_0}
\end{sidewaystable}
\begin{sidewaystable}[p]
    \centering
    \begin{tabularx}{\textwidth}{c *{10}{Y}}
    \toprule
     &
     & \multicolumn{2}{c}{$\gamma=10^2$}
     & \multicolumn{2}{c}{$\gamma=10^3$}
     & \multicolumn{2}{c}{$\gamma=10^4$}
     & \multicolumn{2}{c}{$\gamma=10^5$}\\
    \cmidrule(lr){3-4} \cmidrule(lr){5-6} \cmidrule(lr){7-8} \cmidrule(lr){9-10}
          $\theta$ & $R_0^r=r$ & $\ARL(\mathcal{S}_A^r)$ & $\sqrt{\Var_\infty(\mathcal{S}_A^r)}$ & $\ARL(\mathcal{S}_A^r)$ & $\sqrt{\Var_\infty(\mathcal{S}_A^r)}$ & $\ARL(\mathcal{S}_A^r)$ & $\sqrt{\Var_\infty(\mathcal{S}_A^r)}$ & $\ARL(\mathcal{S}_A^r)$ & $\sqrt{\Var_\infty(\mathcal{S}_A^r)}$\\
    \midrule
        \multirow{4}{*}{0.01}&0&100.27&5.86&1,000.27&176.65&10,000.24&4,565.16&100,000.16&78,476.25\\\cmidrule(lr){2-10}
        &$10^2$&1.08&0.42&900.27&3.81767&9,900.24&4,565.15763&99,900.16&78,476.24\\\cmidrule(lr){2-10}
        &$10^3$&1.0&0.0&4.11&17.51&9,000.24&4,561.56068&99,000.16&78,475.93\\\cmidrule(lr){2-10}
        &$10^4$&1.0&0.0&1.0&0.0&35.49&438.74&90,000.16&78,336.91\\
    \midrule
        \multirow{4}{*}{0.1}&0&100.28&45.92&1,000.28&783.89&10,000.28&9,470.47&99,999.95&99,060.19\\\cmidrule(lr){2-10}
        &$10^2$&4.13&13.58&900.28&782.54&9,900.28&9,470.36&99,899.95&99,060.18\\\cmidrule(lr){2-10}
        &$10^3$&1.0&0.0&35.52&218.63&9,000.28&9,437.96&98,999.95&99,057.08\\\cmidrule(lr){2-10}
        &$10^4$&1.0&0.0&1.0&0.0&349.47&2,525.59&89,999.95&98,606.29\\
    \midrule
        \multirow{4}{*}{0.5}&0&100.44&87.69&1,000.45&973.27&10,000.45&9,956.0&100,000.45&99,937.89\\\cmidrule(lr){2-10}
        &$10^2$&18.11&51.07&900.45313&969.33&9,900.45&9,955.62&99,900.45&99,937.85\\\cmidrule(lr){2-10}
        &$10^3$&1.00001&0.05&173.96&552.33&9,000.44&9,908.34&99,000.45&99,933.15\\\cmidrule(lr){2-10}
        &$10^4$&1.0&0.0&1.0&0.52&1,732.4&5,608.77&90,000.43&99,440.75\\
    \midrule
        \multirow{4}{*}{1.0}&0&100.72&95.72&1,000.13&991.03&10,000.43&9,986.83961&100,000.76&99,982.6\\\cmidrule(lr){2-10}
        &$10^2$&34.92&72.78&899.83&986.41496&9,900.43&9,986.39&99,900.76&99,982.55\\\cmidrule(lr){2-10}
        &$10^3$&1.55&10.15&339.33&745.3&8,997.51&9,937.25&99,000.75&99,977.68\\\cmidrule(lr){2-10}
        &$10^4$&1.0&0.13&6.39&103.18&3,387.5&7,494.37&89,971.64&99,479.67\\
    \bottomrule
    \end{tabularx}
    \caption{Selected values of $\ARL(\mathcal{S}_A^r)$ and $\sqrt{\Var_\infty(\mathcal{S}_A^r)}$ vs. $R_0^r$, $A$, and $\theta$.}
    \label{tab:ARL+Var_vs_R_A_theta}
\end{sidewaystable}

We now proceed to computing the higher moments of the distribution of the GSR procedure's Run-Length. Specifically, consider the GSR procedure's stopping time's standard deviation: $\sqrt{\Var_\infty(\mathcal{S}_A^r)}\triangleq\sqrt{\EV_\infty[(\mathcal{S}_A^r)^2]-(\ARL(\mathcal{S}_A^r))^2}$. Table~\ref{tab:ARL+Var_vs_R_A_theta} presents selected values of $\ARL(\mathcal{S}_A^r)$ and $\sqrt{\Var_\infty(\mathcal{S}_A^r)}$ for different values of $R_0^r$, $A$, and $\theta$. We note that since the GSR procedure's stopping rule is asymptotically (as $A\to\infty$) geometric (see, e.g.,~\citealp{Pollak+Tartakovsky:TVP09,Tartakovsky+etal:IWAP08}), one would expect $\ARL(\mathcal{S}_A^r)$ and $\sqrt{\Var_\infty(\mathcal{S}_A^r)}$ to be close to one another. As may be seen from Table~\ref{tab:ARL+Var_vs_R_A_theta} this is, in fact, the case. However, the asymptotic geometric distribution is attained quicker for lower values of the headstart (i.e., for small $R_0^r$) and higher magnitudes of the change (i.e., for large $\theta$).

We conclude this section with a set of brief analysis of $\Pr_\infty(\mathcal{S}_A>k)$, $k\ge1$, i.e., the survival function of the SR stopping time, and consider various values of $\theta$ and $\ARL(\mathcal{S}_A)\approx\gamma>1$. Recall that that asymptotically (as $\gamma\to\infty$) the $\Pr_\infty$-distribution of $\mathcal{S}_A$ is geometric, and the parameter is the reciprocal of the ARL to false alarm; cf.~\cite{Pollak+Tartakovsky:TVP09,Tartakovsky+etal:IWAP08}. Hence, $\Pr_\infty(\mathcal{S}_A>k)$ should be close to $(1-1/\gamma)^k$ for sufficiently large $\gamma$. We therefore will focus on $\log\Pr_\infty(\mathcal{S}_A>k)$, which as a function of $k$ should be asymptotically (as $A\to\infty$) linear with a slope of $\log(1-1/\gamma)$.

Figure~\ref{fig:logPMF__ARL_100} shows the results for the case of $\gamma=10^2$ for $\theta=0.01,0.1,0.5$ and $1.0$. The dashed line corresponds to the geometric distribution expected in the limit. Note that $\gamma=10^2$ is generally too low of an ARL to false alarm for the asymptotic distribution to kick in. Hence, for $\theta=0.01$ and $0.1$ we observe a substantial deviation of the actual distribution of the GSR stopping time from the limiting geometric one. Yet for $\theta=0.5$ and $1.0$ the distribution is reasonably close to the geometric. Figure~\ref{fig:logPMF__ARL_1000} shows the results for the case of $\gamma=10^3$, also for $\theta=0.01,0.1,0.5$ and $1.0$. In this case the distribution for $\theta=0.5$ and $1.0$ is almost indistinguishable from geometric, and that for $\theta=0.1$ is much closer to being geometric than when $\gamma=10^2$; that said, $\theta=0.01$ is still a problem. Figure~\ref{fig:logPMF__ARL_10000} corresponds to the case when $\gamma=10^4$. In this case the distribution corresponding to $\theta=0.1$ is almost geometric, and that corresponding to $\theta=0.01$ is still considerably different from the limiting geometric. One conclusion to draw from this analysis is that the rate of convergence to the geometric distribution heavily depends on the magnitude of the change. For faint changes ($\theta=0.01$ and $0.1$), the rate is extremely slow, even if the ARL to false alarm is large. However, for moderate and contrast changes ($\theta=0.5$ and $1.0$), the distribution is very close to geometric already for $\gamma=10^2$.
\begin{sidewaysfigure}[p]
%\begin{figure}
\begin{subfigure}[b]{0.35\textheight}
    \centering
    \includegraphics[width=0.35\textheight]{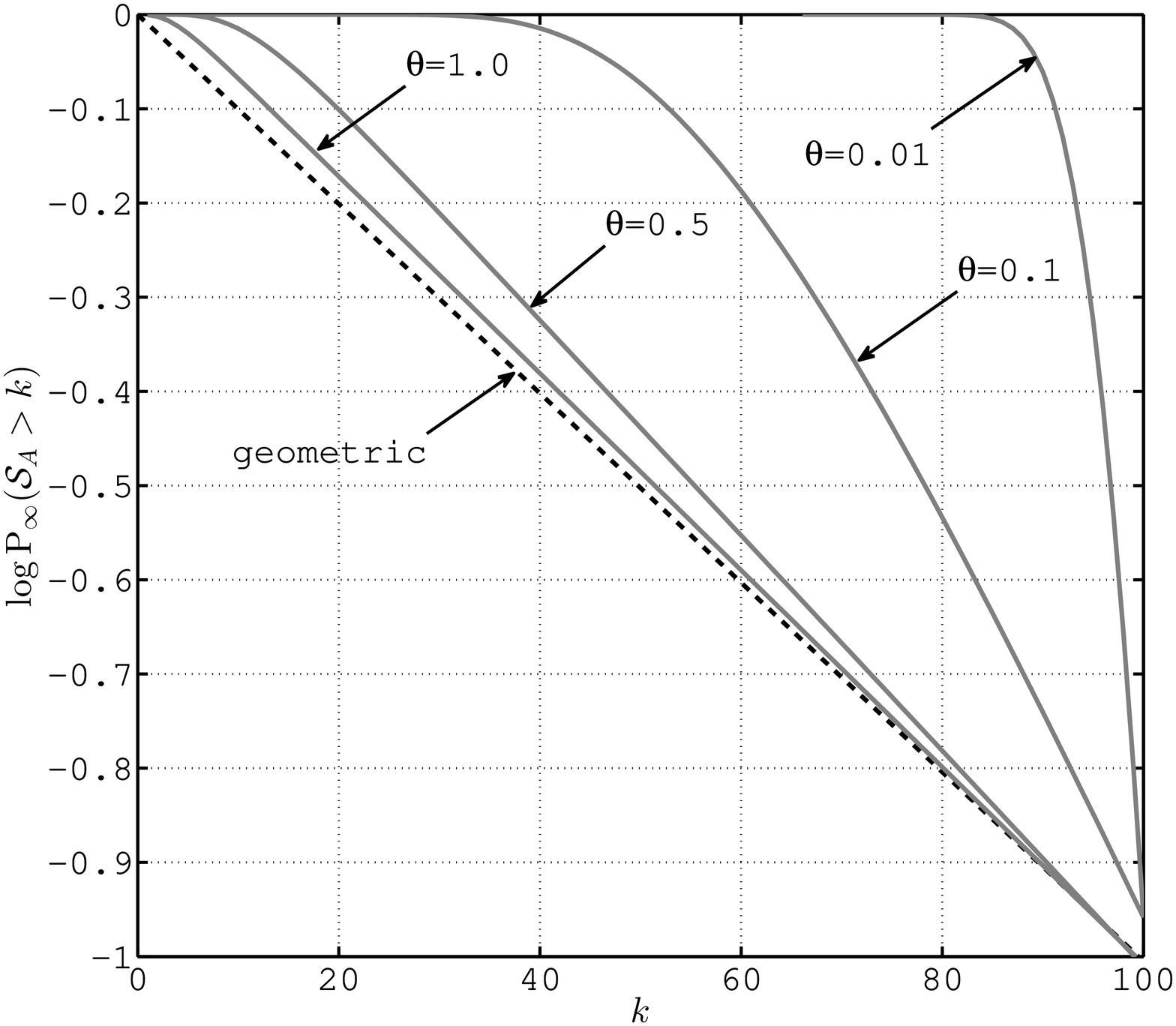}
    \caption{$\gamma=10^2$.}
    \label{fig:logPMF__ARL_100}
\end{subfigure}
\begin{subfigure}[b]{0.35\textheight}
    \centering
    \includegraphics[width=0.35\textheight]{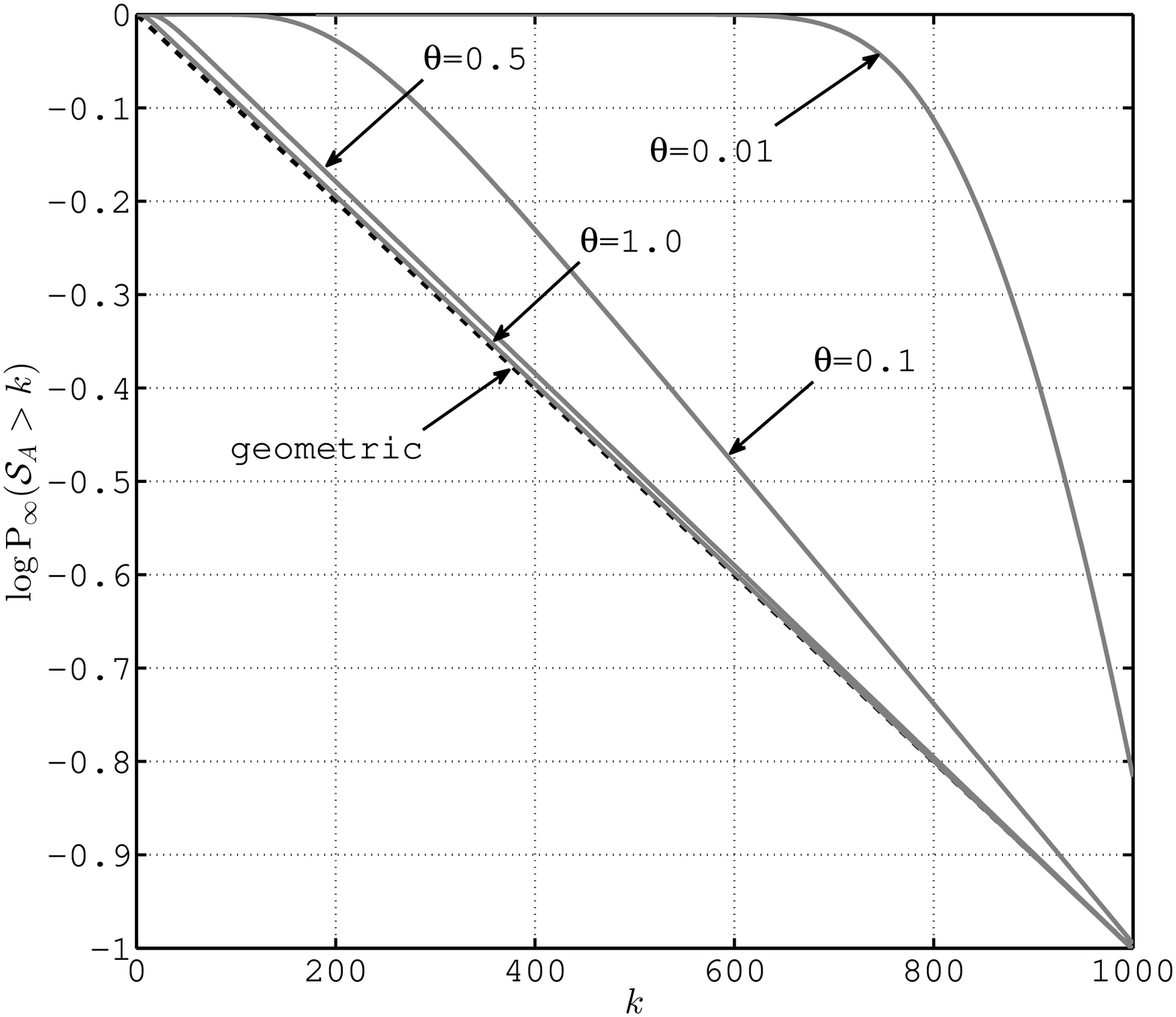}
    \caption{$\gamma=10^3$.}
    \label{fig:logPMF__ARL_1000}
\end{subfigure}
\begin{subfigure}[b]{0.35\textheight}
    \centering
    \includegraphics[width=0.35\textheight]{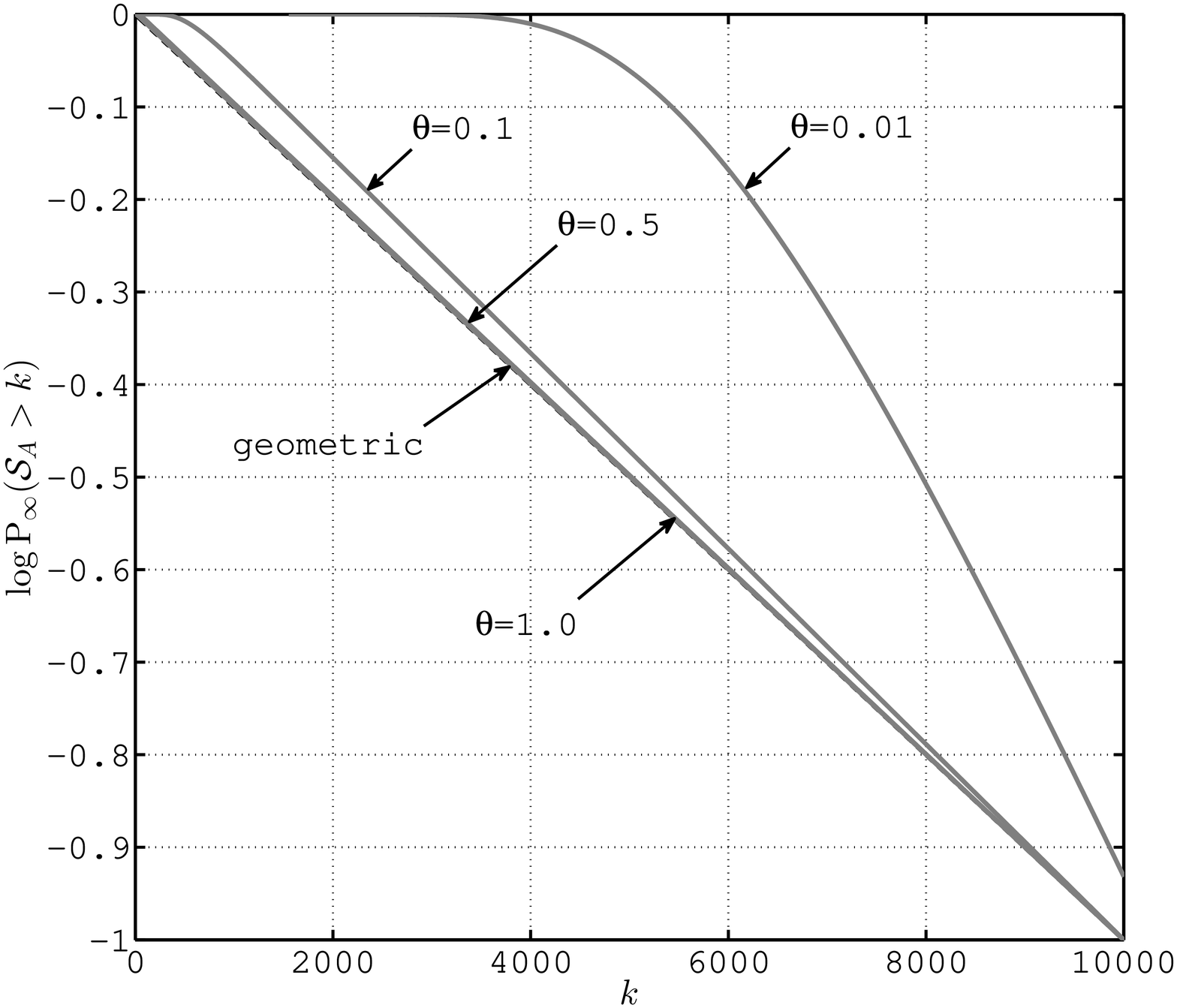}
    \caption{$\gamma=10^4$.}
    \label{fig:logPMF__ARL_10000}
\end{subfigure}
%\end{figure}
\caption{$\log\Pr(\mathcal{S}_A>k)$ as a function of $k=1,2,\ldots,\gamma$.}
\end{sidewaysfigure}

%-------------------------------------------------------------------------------------------------%
\section{Further discussion}
\label{sec:further-discussion}

The proposed numerical method (including the change-of-measure trick) can be extended to other performance measures as well as to other detection procedures. For example, note that throughout the paper we used the usual ARL to false alarm as a capture of the risk of false detection. An alternative, more exhaustive measure of he risk of false detection is based on the ``worst'' $\Pr_\infty$-probability of sounding a false alarm (PFA) within a batch of $m\ge1$ successive observations indexed $k+1$ through $k+m$ (inclusive), where $k\ge0$. To define this probability, one can either use $\sup_{0\le k}\Pr_\infty(k<\T\le k+m)$ or $\sup_{0\le k}\Pr_\infty(k<\T\le k+m|\T> k)$; see, e.g.,~\cite{Lai:IEEE-IT1998}. The constraint $\ARL(\T)\ge\gamma$, where $\gamma>1$, is then replaced with $\sup_{0\le k}\Pr_\infty(k<\T\le k+m)<\alpha$ or $\sup_{0\le k}\Pr_\infty(k<\T\le k+m|\T>k)<\alpha$, respectively, where $\alpha\in(0,1)$. As shown by~\cite{Lai:IEEE-IT1998}, either of these new constraints is stronger than the original $\ARL(\T)\ge\gamma$, assuming a properly picked $m\ge1$. Following~\cite{Tartakovsky:IEEE-CDC05,Tartakovsky:SA08-discussion}, let us focus on the conditional PFA and denote $\PFA_k^m(\T)\triangleq\Pr_\infty(k<\T\le k+m|\T>k)$. Our numerical method can be used to evaluate this probability for the SR--$r$ procedure using
\begin{align*}
\PFA_k^m(\mathcal{S}_A^x)
&=
1-\frac{\rho_{k+m}(x,A)}{\rho_k(x,A)},
\end{align*}
where $\rho_0(x,A)\equiv1$ for all $x,A\in\mathbb{R}$, and $\rho_k(x,A)$, for $k=1,2\ldots$ are given by~\eqref{eq:recursion2}.

We now show how to generalize the method to other procedures. Consider a generic detection procedure described by the stopping time $\mathcal{T}_A^s=\inf\{n\ge1\colon V_n^s\ge A\}$
where $A>0$ is a detection threshold, and $\{V_n^s\}_{n\ge0}$ is a generic Markovian detection statistic that admits the recursion $V_{n}^s=\psi(V_{n-1}^s)\LR_{n}$ for $n=1,2,\ldots$ with $V_0^s=s$, where $\psi(x)$ is a non-negative-valued function, and the headstart $s$ is fixed. This generic stopping time was considered by~\cite{Moustakides+etal:SS11} and it describes a rather broad class of detection procedures; in particular, the GSR procedure given by~\eqref{eq:T-SRr-def} and~\eqref{eq:statistic-SRr-def} is a member of this class with $\psi(x)=1+x$. Also, for the CUSUM chart (not considered here) it is enough to take $\psi(x)=\max\{1,x\}$. This universality of $\mathcal{T}_A^s$ enables one to evaluate practically any procedure that is a special case of $\mathcal{T}_A^s$, merely by picking the right $\psi(x)$.

The integral equations framework (including the numerical method) developed above can be easily extended to $\mathcal{T}_A^s$. It suffices to let
\begin{align*}
\mathcal{K}_d(x,y)
&=
\dfrac{\partial}{\partial y}\Pr_d(V_{n+1}^s\le y|V_n^s=x)
=
\dfrac{\partial}{\partial y}P_d^{\LR}\left(\frac{y}{\psi(x)}\right),
\; d=\{0,\infty\}
\end{align*}
to denote the transition probability density kernel for the Markov process $\{V_n^s\}_{n\ge1}$. Note that $dP_0^{\LR}(t)=tdP_{\infty}^{\LR}(t)$ is still true, and therefore, so is $\psi(x)K_0(x,y)=yK_{\infty}(x,y)$. This is a generalization of the $(1+x)K_0=yK_\infty(x,y)$ identity used earlier for the GSR procedure.

However, in general, while the change-of-measure trick can be used to develop numerical methods to evaluate other detection procedures, the improved convergence rate will not take effect as fast as for the GSR procedure. This is because of the GSR procedure's martingale property that the numerical method exploits.

%-------------------------------------------------------------------------------------------------%
\section{Conclusion}
\label{sec:conclusion}

We proposed a numerical method to compute the Generalized Shiryaev--Roberts (GSR) procedure's pre-change Run-Length distribution and its moments. The GSR procedure is an umbrella term for the original SR procedure and its extension -- the SR--$r$ procedure. The proposed method is based on the integral-equations approach, and uses the collocation framework with the basis functions selected so as to exploit a certain change-of-measure gimmick and a specific martingale property of the GSR procedure's detection statistic. This resulted in substantial improvement of the method's accuracy and overall stability. We also carried out a complete accuracy analysis of the method, and provided a tight upper bound on the method's error as well as showed the theoretical rate of convergence to be quadratic. By testing the method in a specific change-point scenario we confirmed that the accuracy is high for a wide range of changes (faint, moderate, contrast) and for a large range of values of the ARL to false alarm. We also confirmed that the method's accuracy remains high even if the partition is rough.

%-------------------------------------------------------------------------------------------------%
\section*{Acknowledgements}

The authors are thankful to Shelemyahu~Zacks (SUNY Binghamton), Shyamal~K.~De (SUNY Binghamton), Sven~Knoth (Helmut Schimdt University, Hamburg, Germany), and to George~V.~Moustakides (University of Patras, Greece) for reading this work's preliminary draft and for providing valuable feedback that helped improve the paper. The authors are also grateful to the Editor-in-Chief, Nitis~Mukhopadhyay (University of Connecticut-Storrs), and to the anonymous reviewers whose comments helped to improve the manuscript further.

%-------------------------------------------------------------------------------------------------%
%\bibliographystyle{sqa}
%\bibliography{main,integral-equations,spc,operator-theory,size-bias}

\end{document}